\documentclass[aos,preprint]{imsart}

\RequirePackage[OT1]{fontenc}
\RequirePackage{amsthm,amsmath}
\RequirePackage[numbers,sort]{natbib}

\usepackage{fancyhdr}
\usepackage{amsfonts,amssymb}
\usepackage{epsfig,graphicx}
\usepackage{epsf} 
\usepackage{graphics} 
\usepackage{subcaption}
\usepackage{booktabs}

\def\indist{\rightsquigarrow}
\def\ind{\perp\!\!\!\perp}
\def\T{{ \mathrm{\scriptscriptstyle T} }}

\newcommand{\var}{\text{var}}
\newcommand{\cov}{\text{cov}}
\newcommand{\Pb}{\mathbb{P}}
\newcommand{\Pn}{\mathbb{P}_n}

\newcommand{\E}{\mathbb{E}}
\newcommand{\R}{\mathbb{R}}

\newcommand{\bO}{\mathbf{O}}
\newcommand{\bo}{\mathbf{o}}
\newcommand{\bV}{\mathbf{V}}
\newcommand{\bv}{\mathbf{v}}
\newcommand{\bX}{\mathbf{X}}
\newcommand{\bx}{\mathbf{x}}

\def\logit{\text{logit}}
\def\expit{\text{expit}}

\DeclareMathOperator*{\argmin}{arg\,min}
\DeclareMathOperator*{\argmax}{arg\,max}

\DeclareSymbolFont{bbold}{U}{bbold}{m}{n}
\DeclareSymbolFontAlphabet{\mathbbold}{bbold}
\newcommand{\one}{\mathbbold{1}}
\newtheorem{theorem}{Theorem}
\newtheorem{lemma}{Lemma}
\newtheorem{corollary}{Corollary}
\newtheorem{proposition}{Proposition}
\theoremstyle{definition}
\newtheorem{definition}{Definition}
\newtheorem{property}{Property}
\theoremstyle{remark}
\newtheorem{assumption}{Assumption}

\newtheorem{remark}{Remark}

\def\spacingset#1{\renewcommand{\baselinestretch}%
{#1}\small\normalsize} \spacingset{1}

\arxiv{arXiv:1801.03635}

\startlocaldefs
\numberwithin{equation}{section}
\theoremstyle{plain}

\endlocaldefs

\begin{document}

\begin{frontmatter}
\title{Sharp Instruments for Classifying Compliers and Generalizing Causal Effects}
\runtitle{Sharp Instruments}

\begin{aug}
\author{\fnms{Edward H.} \snm{Kennedy}\ead[label=e1]{edward@stat.cmu.edu}},
\author{\fnms{Sivaraman} \snm{Balakrishnan}\ead[label=e2]{siva@stat.cmu.edu}},
\author{\fnms{Max} \snm{G'Sell}\ead[label=e3]{mgsell@stat.cmu.edu}}

\thankstext{t1}{The authors thank the Johns Hopkins Causal Inference Working Group, Luke Keele, Betsy Ogburn, Jamie Robins, Fredrik S{\"a}vje, and Larry Wasserman for helpful comments and discussions. This work was partially supported by the NSF grant DMS-17130003.}
\runauthor{E.H. Kennedy et al.}

\affiliation{Carnegie Mellon University}

\address{Department of Statistics and Data Science \\
Carnegie Mellon University\\
Baker Hall 132 \\
Pittsburgh, Pennsylvania 15213 USA \\
\printead{e1} \\
\phantom{E-mail:\ }\printead*{e2} \\
\phantom{E-mail:\ }\printead*{e3}}

\end{aug}

\begin{abstract}
It is well-known that, without restricting treatment effect heterogeneity, instrumental variable (IV) methods only identify ``local'' effects among compliers, i.e., those subjects who take treatment only when encouraged by the IV. Local effects are controversial since they seem to only apply to an unidentified subgroup; this has led many to denounce these effects as having little policy relevance. However, we show that such pessimism is not always warranted: it can be possible to accurately predict who compliers are, and obtain tight bounds on more generalizable effects in identifiable subgroups. We propose methods for doing so and study estimation error and asymptotic properties, showing that these tasks can sometimes be accomplished even with very weak IVs. We go on to introduce a new measure of IV quality called ``sharpness'', which reflects the variation in compliance explained by covariates, and captures how well one can identify compliers and obtain tight bounds on identifiable subgroup effects. We develop an estimator of sharpness, and show that it is asymptotically efficient under weak conditions. Finally we explore finite-sample properties via simulation, and apply the methods to study canvassing effects on voter turnout. We propose that sharpness should be presented alongside strength to assess IV quality.
\end{abstract}



\end{frontmatter}

\section{Introduction}
\label{sec:intro}
Instrumental variable (IV) methods are a widespread tool for identifying causal effects in studies where treatment is subject to unmeasured confounding. These methods have been used in econometrics since the 1920s \citep{wright1934method}, but have only been set within a formal potential outcome framework more recently \citep{robins1989analysis, manski1990nonparametric, imbens1994identification}. Roughly speaking, an instrument is a variable that is associated with treatment, but is itself unconfounded and does not directly affect outcomes. An archetypal example is in randomized experiments with noncompliance, where initial randomization can be an instrument for the treatment that was actually received. IV methods are also used widely in observational studies, where investigators try to exploit natural randomness in, for example, treatment preference, distance, or  time. We refer to  \citet{hernan2006instruments, imbens2014instrumental, baiocchi2014instrumental} for a more comprehensive review and examples. 

Despite their popularity and prevalence, instrument variable methods bring some difficulties that do not arise in studies of unconfounded treatments. In particular, without restricting treatment effect heterogeneity in some way or adding extra assumptions, one cannot identify average treatment effects across the entire population. For example, even in the simplest setting involving a randomized study with one-sided noncompliance (e.g., where subjects randomized to control cannot access treatment), the treatment effect is nonparametrically identified only among those who actually receive treatment.

One option then is to pursue bounds on the overall average treatment effect \citep{robins1989analysis, manski1990nonparametric, balke1997bounds}. This approach is robust, but has been criticized on the grounds that the resulting inferences can be so imprecise that they are not helpful for making policy decisions. Others argue that even wide bounds are useful, by making explicit that any more precision would require further assumptions \citep{robins1996identification}. An alternative  approach incorporates extra assumptions to achieve point identification. Classically this was often accomplished via constant treatment effect assumptions within linear structural equation models. More recent generalizations allow for heterogeneous treatment effects and non-linear models based on weaker homogeneity restrictions, e.g., no effect modification by the instrument, or other no-interaction or parametric assumptions  \citep{robins1994correcting,tan2010marginal}. However, as noted by \citet{tchetgen2013alternative}, parametric identification can be problematic since it a priori restricts the effect of interest, and such functional form knowledge is not typically available in practice.

Yet another strategy instead assumes monotonicity \citep{robins1989analysis, imbens1994identification}, which rules out the possibility that the instrument could encourage someone to take control when they would otherwise take treatment (i.e., rules out so-called defiers).  This approach is unique in allowing nonparametric identification of a causal effect, but only a local effect among the subgroup of compliers, i.e., those subjects who would only take treatment when encouraged by the instrument \citep{imbens1994identification, angrist1996identification}. These local average treatment effects (LATEs) have generated some controversy, since they are defined in an unidentified subgroup that is not directly observed; we refer to \citet{imbens2014instrumental} and \citet{swanson2014think} for a recent debate. The issue is that, for encouraged subjects, we never get to see whether they would have taken treatment if not encouraged, and vice versa for unencouraged subjects. Therefore it is generally unknown whether any given subject is a complier or not. 

One justification for continuing to pursue complier effects is that they allow something causal to be learned in broken or ``second-best'' studies with unmeasured confounding, even without restricting effect heterogeneity \citep{imbens2010better, imbens2014instrumental}. In other words, although complier effects may not be an ideal target estimand, in reality most observational studies are confounded and so the ideal is not attainable. Despite this, one might argue, complier effects can still reveal a piece of the puzzle of the causal structure, and can in principle be used together with bounds on more standard effects.

However, such justification is not always convincing, yielding some lively debate. \citet{robins1996identification} stressed early on that the complier subgroup is not identified, and gave examples where complier effects are not of primary policy interest. \citet{pearl2009causality} says the complier ``subpopulation cannot be identified and, more seriously, it cannot serve as a basis for policies.'' \citet{deaton2010instruments} compares targeting local effects to the drunk who only looks for his keys near the lamppost, since that is where the light is. \citet{swanson2014think} state that complier effects ``only pertain to an unknown subset of the population'', and that ``as we do not know who is a complier, we do not know to whom our new policy should apply.'' These kinds of critiques suggest that generalization via  complier effects is a hopeless endeavor. In this paper, we explore whether this is necessarily the case.

\subsection{Motivating Example}

The most common way to judge an instrument's quality is by its strength, typically defined as the proportion of compliers  $\Pb(C=1)$ \citep{baiocchi2014instrumental}, where $C$ is the unobserved indicator of complier status.  However, consider  Figure 1. 

\begin{figure}[h!] 
\centering
\includegraphics[width=.64\textwidth]{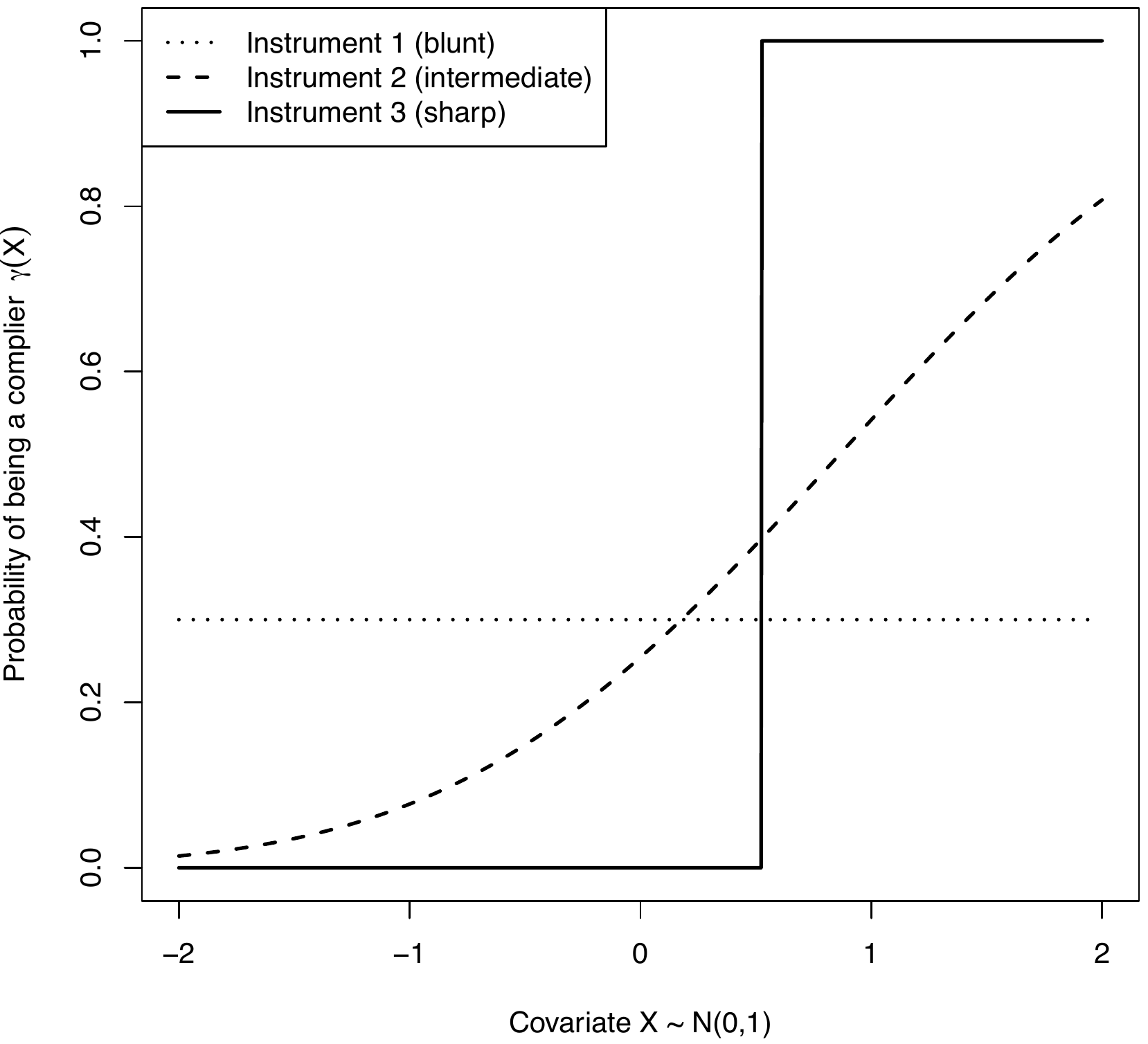} 
\caption{Compliance probability $\gamma(x)=\Pb(C=1 \mid X=x)$ for three equally strong IVs.} \label{fig:explot}
\end{figure}

In this toy example, there is a single covariate $X \sim N(0,1)$ and three candidate instruments, $(Z_1, Z_2, Z_3)$. All three instruments have exactly the same strength, each yielding $30\%$ compliers in the population. However, the available information about compliers changes drastically across the three cases. For the first instrument $Z_1$, it is only known that the probability of compliance is 30\% for each subject, regardless of covariate value. Thus there is no additional information beyond the marginal strength; this is the worst-case setup often considered in critiques of complier-specific effects. However, consider the third instrument $Z_3$. For this instrument, the covariate $X$ perfectly predicts compliance, so that 
$C=\one( X>0.5244\ldots )$ 
and the complier-specific effect 
$$ \E(Y^{a=1} - Y^{a=0} \mid C=1 ) = \E(Y^{a=1} - Y^{a=0} \mid X > 0.5244\ldots) $$
is in fact just a conditional effect within an observable subgroup. Therefore when using $Z_3$ as an instrument all aforementioned concerns about local effects fall away completely. Importantly, this fact is not reflected at all in the strength of the instrument. It is also missed by the first-stage F-statistic, another common measure of instrument quality \citep{bound1995problems, staiger1997instrumental}, regardless of whether modeling assumptions are correct or not; we provide a simulated example in Appendix~\ref{app:fstat}. The second instrument $Z_2$ is an intermediate between $Z_1$ and $Z_3$. 

This example raises many interesting questions, which arise more generally in any instrument variable study. How can we quantify the extra information afforded by instruments like $Z_2$ relative to $Z_1$? Can we leverage this information to obtain more accurate guesses of who the compliers are? Can this help us go beyond local effects and instead identify effects in observable subgroups? The goal of this paper is to provide answers to these questions. Overall, we find that pessimism about local effects may be warranted in studies with blunt instruments. However, our work indicates that many concerns can be ameliorated or avoided in studies with sharp instruments, even if they are weak.

\subsection{Outline \& Contributions}
In this paper we characterize sharp instruments as those that admit accurate complier predictions, and tight bounds on effects in identifiable subgroups. We present some notation and our assumptions in Section~\ref{sec:not}. In Section~\ref{sec:class} we discuss the problem of classifying compliers. We propose several complier classification rules, derive their large-sample errors, and discuss optimality and estimation.  In Section~\ref{sec:bounds} we discuss using instruments to bound effects in identifiable subgroups, characterize the subgroup that yields tightest bounds, and propose corresponding estimators for these bounds. In Section~\ref{sec:sharp} we propose a new summary measure of instrument quality called sharpness, which is separate from strength, and measures the variation in compliance explained by compliance scores. We show that sharper instruments yield better identification of compliers and tighter bounds on effects in identifiable subgroups, and present an efficient nonparametric estimator of the sharpness of an instrument. Our estimators are based on influence functions so as to yield fast convergence rates and tractable inference even when relying on modern flexible regression methods;  all methods are implemented in the \verb|npcausal| package in R. Finally, in Section~\ref{sec:sims} we study finite-sample properties via simulation, and apply our methods in a study of effects of canvassing on voter turnout \citep{green2003getting}.

\section{Notation \& Setup}
\label{sec:not}

We consider the usual instrumental variable setup, where one observes an iid sample 
$\{\bO_1,\ldots,\bO_n\} \sim \Pb$ with
$$ \bO = (\bX,Z,A,Y) $$
for  covariates $\bX \in \mathcal{X} \subseteq \R^p$, a binary instrument $Z \in \{0,1\}$, a binary treatment $A \in \{0,1\}$, and some outcome $Y \in [0,1]$ of interest. We let $Y^a$ denote the potential outcome \citep{rubin1974estimating} that would have been observed had treatment been set to $A=a$, and the goal is to learn about the distribution of the treatment effect $Y^{a=1}-Y^{a=0}$. We also need to define potential outcomes under interventions on the instrument. Thus let $Y^{za}$ denote the potential outcome that would have been observed under both $Z=z$ and $A=a$, and similarly let $A^z$ and $Y^z=Y^{zA^z}$ denote the potential treatment and outcome when the instrument is set to $Z=z$. In the statement of some of our results we use the standard statistical big-O notation, as well as the shorthand $a \lesssim b$ to denote $a \leq Cb$ for some universal positive constant $C > 0$.

To ease the presentation we let
\begin{align*}
\pi_z(\bx) = \Pb(Z=z \mid \bX=\bx) , \
\lambda_z(\bx) = \Pb(A=1 \mid \bX=\bx, Z=z) , \
\end{align*}
denote the instrument propensity score and treatment regression,
and let 
\begin{align*}
\gamma(\bx) = \lambda_1(\bx) - \lambda_0(\bx)
\end{align*} 
denote the corresponding IV-difference. 

We  let $C=\one(A^{z=1}>A^{z=0})$ denote the latent variable indicating whether a subject is a complier, i.e., whether a subject would respond to encouragement by the instrument. As mentioned in Section~\ref{sec:intro}, $C$ is not directly observed. Nonetheless, it is well-known \citep{angrist1996identification, abadie2003semiparametric, hernan2006instruments}  that causal effects among compliers are nonparametrically identified under the following assumptions:
\begin{assumption}[Consistency] $A=ZA^{z=1} + (1-Z)A^{z=0}$ and $Y=ZY^{z=1} + (1-Z)Y^{z=0}$.
\end{assumption}
\begin{assumption}[Positivity] $\Pb\{ \epsilon \leq \pi_z(\bX) \leq 1-\epsilon \}=1$ for some $\epsilon>0$.
\end{assumption}
\begin{assumption}[Unconfounded IV] $Z \ind (A^z, Y^z) \mid \bX$.
\end{assumption}
\begin{assumption}[Exclusion Restriction] $Y^{za}=Y^a$.
\end{assumption}
\begin{assumption}[Strong Monotonicity] $\Pb(A^{z=1}<A^{z=0})=0$ and $\Pb(C=1) \geq \epsilon>0$.
\end{assumption}
(Note the lower-case indices $z,a$ represent arbitrary values of the instrument and treatment). We refer elsewhere  \citep{angrist1996identification, abadie2003semiparametric, hernan2006instruments} for a detailed discussion of the above assumptions, which are standard in the literature (as mentioned in Section~\ref{sec:intro}, monotonicity is sometimes replaced by effect homogeneity or no-interaction assumptions). Assumptions 1--5 imply that the average effect among compliers (called the local average treatment effect, or LATE) with $\bV=\bv$ (for any subset $\bV \subseteq \bX$) is given by
\begin{equation}
\label{eq:late}
\E(Y^{a=1} - Y^{a=0} \mid \bV, C=1) =   \frac{ \E\{ \E(Y \mid \bX,Z=1) - \E(Y \mid \bX,Z=0) \mid \bV \} }{ \E\{ \E(A \mid \bX,Z=1) - \E(A \mid \bX,Z=0) \mid \bV \} }  . 
\end{equation}
This is the kind of local effect discussed in Section~\ref{sec:intro}. Crucially, Assumptions 1--3 and 5 also imply that the chance of being a complier given covariates is given by 
$$\Pb(C=1 \mid \bX=\bx)=\gamma(\bx)$$ 
and so strength is given by $\mu \equiv \Pb(C=1)=\E\{ \gamma(\bX)\}$. The function $\gamma(\bx)$ has been termed the ``compliance score''  \citep{follmann2000effect, joffe2003compliance, aronow2013beyond}, and is an example of a ``principal score'' \citep{jo2009use, stuart2015assessing, ding2017principal, feller2017principal}. Note that the principal score literature typically assumes independence between principal strata indicators (e.g., $C$) and potential outcomes, which we avoid here. 

\section{Classifying Compliers}
\label{sec:class}

Heuristically, we propose calling instruments sharp when it is possible to predict compliance well, and obtain tight bounds on effects in identifiable subgroups. In this section we discuss the first of these properties, i.e., that of predicting the latent complier status $C$ based on observed covariate information $\bX$. We present several complier classification rules, characterize their errors and the relations between them, and discuss optimality. Finally we present corresponding estimators, and discuss estimation error and large-sample properties.

\begin{remark}
Our view is that complier classification can be a valuable tool in practice, complementary to assessing compliance scores $\gamma$ on their own. A first reason why is pragmatic: it may be simply preferred (e.g., based on ease of interpretation) for practitioners to inspect a concrete set of likely compliers. Also, as we will discuss shortly, there is one particular classifier whose predicted compliers can act as surrogates for estimating any complier characteristic. Another pragmatic justification is that, statistically, complier classification is at least as easy as  compliance score estimation: as in standard classification, one's score estimates could be severely biased and yet good classification error might still be attainable. For a trivial example, suppose $\gamma=0$ for all $\bx$ so there are no compliers, but estimated compliance scores $\widehat\gamma = 0.4$ everywhere and so are highly biased; even so, the classifier $\widehat{h}=\one(\widehat\gamma>t)$ is perfectly accurate for all $t \geq 0.4$. 

Importantly, classification is also particularly crucial whenever decision-making is required. For example, from a policy perspective, encouraging non-compliers may be wasted effort since non-compliers will by definition have the same behavior regardless of encouragement. Thus one could consider the following two-stage treatment policy: first compliance status is predicted, and then treatment is recommended only to those predicted compliers who are expected to benefit. Complier classification could also be useful for simultaneously minimizing non-compliance and increasing generalizability in experiments: for example  one could run a doubly randomized preference trial \citep{marcus2012estimating} where those subjects who are predicted to be compliers are randomized to the experimental arm with a higher probability, whereas predicted non-compliers are randomized to the observational arm with a higher probability. We aim to explore the use of complier classification in these specific decision-making contexts in detail in future work.
\end{remark}

\subsection{Classifiers \& Properties} \label{sec:classprop}
As noted earlier, although compliers are not strictly identified it is possible to predict compliance status based on the fact that 
Assumptions 1--3 and 5 suffice to ensure that
$$ \Pb(C=1 \mid \bX=\bx) = \gamma(\bx).$$
As stated in the following proposition, we can similarly identify the classification error $\mathcal{E}(h)=\Pb(C \neq h)$ for any given complier classification rule $h$, which we define as an arbitrary measurable function $h: \mathcal{X} \mapsto \{0,1\}$ mapping the covariates to a binary prediction.  As discussed further following~\eqref{eq:h2}
this proposition and subsequent results generalize in a natural way to classifiers that are stochastic.

\begin{proposition} \label{errid}
For any complier classification rule $h: \mathcal{X} \mapsto \{0,1\}$, the corresponding classification error $\mathcal{E}(h) = \Pb\{C \neq h(\bX)\}$ is identified under Assumptions 1--3 and 5 as
$$ \mathcal{E}(h) = \E\Big[ \gamma(\bX) \{1- h(\bX) \} + \{1- \gamma(\bX) \} h(\bX) \Big] . $$
\end{proposition}

\noindent A proof of Proposition \ref{errid} and all other results can be found in the Appendix. Although the compliance score has been discussed in the literature since at least  \citet{follmann2000effect} and \citet{joffe2003compliance}, we have not seen it used before for the specific purpose of predicting who the compliers are, nor have we seen any discussion of the error of this task. In contrast, most work seems to focus on the related but separate problem of estimating complier characteristics, such as $\E(\bX \mid C=1)$ \citep{abadie2003semiparametric, baiocchi2014instrumental}. As explained above, we feel compliance classification is practically important and yet under-studied, particularly for so-called sharp instruments that allow for accurate prediction. If the error $\mathcal{E}(h)$ can be made small, then it is possible to know who the compliers are quite precisely. A main point of this paper is to formalize this, and show that it is possible for compliers to be accurately classified even with weak instruments. 

The optimal classifier $h_0$ in terms of minimizing the error $\mathcal{E}(h)$ is given by the Bayes decision function
\begin{equation}
\argmin_{h: \mathcal{X} \mapsto \{0,1\}} \mathcal{E}(h) = \one\{ \gamma(\bx) > 1/2\} \equiv h_0(\bx) .
\end{equation}
The proof of this fact follows from the same logic as in standard classification problems \citep{devroye1996probabilistic}. Shortly we will discuss estimation of the Bayes decision via the plug-in estimator $\one(\widehat\gamma > 1/2)$. One could also consider empirical risk minimizers of the form
$$ \widehat{h} = \argmin_{h \in \mathcal{H}} \widehat{\mathcal{E}}(h) $$
for an appropriate class $\mathcal{H}$ (e.g., linear classifiers) and estimator $\widehat{\mathcal{E}}(h)$ of the error. We leave this to future work, only considering plug-in 
classifiers in this paper. 

Despite its simplicity and optimality (with respect to classification error), the rule $h_0$ may have some practically unfavorable properties in the setting of complier classification. In particular, the set of putative compliers returned by $h_0$ could have a very different size compared to the true set. We call classifiers strength-calibrated if they output sets with the same size as the true set.

\begin{property}
A complier classification rule $h: \mathcal{X} \mapsto \{0,1\}$ is \textit{strength-calibrated} if
\begin{equation}
\Pb\{h(\bX)=1\} = \Pb(C=1) . \tag{P1}
\end{equation}
\end{property}
\noindent If for no other reason, strength calibration can be important in complier classification simply because strength $\mu=\Pb(C=1)$ is such a fundamental quantity in instrumental variable problems. Strength is often the primary criterion used to judge instrument quality, since the more compliers there are, the more subjects there are for whom the local effect is relevant, and so the more meaningful and generalizable the effect is. Thus one might prefer to trade off some error for a classification rule that accurately reflects the underlying size of the complier population, for instance, in settings where achieving a minimum error threshold is sufficient, rather than precise minimization.

Similarly, it is possible that the optimal rule $h_0$ would never guess any compliers (i.e., $h_0=0$ with probability one), which could be unfavorable for a practical analysis. For example, suppose $\gamma=\mu=49\%$, or that the covariate $X$ was uniform and $\gamma(x)=x/2$. Then the optimal rule $h_0$ would return the empty set in both cases, even though the proportion of compliers is nearly one-half and a quarter, respectively. The empty set could be an unsatisfying result for a practitioner who was curious about identifying which particular subjects were compliers. 

A simple strength-calibrated rule is given by the quantile-threshold classifier
\begin{equation} \label{eq:quantile}
h_q(\bx) = \one\{ \gamma(\bx) > q \}
\end{equation}
where $q = F^{-1}(1-\mu)$, and $F(t) = \Pb\{ \gamma(\bX) \leq t\}$ is the cumulative distribution function of the compliance score. The rule $h_q$ simply predicts that the 100$\mu$\% of subjects with the highest compliance scores are the compliers. That $h_q$ is strength-calibrated follows since 
$$ \Pb(h_q=1) = \Pb\{F(\gamma) > 1-\mu \} = 1-(1-\mu)$$
because $F(\gamma)$ follows a uniform distribution. Here we have assumed there exists an exact (unique) quantile $q$ such that $F(q) = 1-\mu$; when this does not hold, one could instead enforce a weaker condition like $\Pb(h=1) \geq \Pb(C=1)$. In the next subsection we show that, when there is a unique quantile, no other strength-calibrated rule can achieve a better classification error than $h_q$. 

One could similarly consider rules of the form $h_t(\bx) = \one\{ \gamma(\bx) > t \}$ for a generic $t \in [0,q]$, if a finer trade-off between classification error and size is required, e.g., if the increase in classification error when moving from $h_0$ to $h_q$ is too severe. 

Another restriction that may be useful to consider in complier classification problems is that of ensuring the covariate distributions among the predicted and true compliers are the same. We call this distribution-matching.

\begin{property}
A complier classification rule $h: \mathcal{X} \mapsto \{0,1\}$ is \textit{distribution-matched} if
\begin{equation}
\Pb\{ \bX \leq \bx \mid h(\bX) = 1 \} = \Pb(\bX \leq \bx \mid C=1) \ \forall \bx . \tag{P2}
\end{equation}
\end{property}
\noindent Distribution matching is useful as it allows practitioners to query the covariate distribution among predicted compliers to learn about the true complier distribution. This provides a user-friendly method for assessing complier characteristics, which can be an alternative to direct estimation via the identifying expressions given for example by \citet{abadie2003semiparametric}. Strength-calibration and distribution-matching together imply that $\Pb(\bX \leq \bx \mid h=0) = \Pb(\bX \leq \bx \mid C=0)$, so the statistician can also estimate prevalence ratios \citep{baiocchi2014instrumental} like $\Pb(\bX \leq \bx \mid C=1)/\Pb(\bX \leq \bx)$ by simply comparing  predicted compliers to the whole sample, i.e., by estimating $\Pb(\bX \leq \bx \mid h=1)/\Pb(\bX \leq \bx)$  for a distribution-matched classifier $h$.

In fact, we show in the next subsection that the only  rule that is both strength-calibrated and distribution-matched is the stochastic classifier
\begin{equation}
h_s(\bx) = \one\{ \gamma(\bx) > U \} \sim \text{Bernoulli}\{\gamma(\bx)\} \label{eq:h2}
\end{equation}
where $U \sim \text{Unif}(0,1)$ is an independent draw from the uniform distribution on [0,1].  Note that $h_s$ randomly predicts that a subject with covariates $\bx$ is a complier with probability $\gamma(\bx)$. To be precise, since $h_s$ is stochastic it should really also be indexed by $U$, as in $h_s(\bx)=h_s(\bx,U)$. It is implicit that any expectations $\E(h)=\E\{h(\bX,U)\}$ are over both $\bX$ and $U$. 

\subsection{Classifier Errors \& Relations}

 In the following results, we characterize the errors of the classifiers $h_q$ and $h_s$, show that they are optimal in the classes of strength-calibrated and distribution-matched classifiers, respectively, and relate their error to the minimal Bayes error $\mathcal{E}(h_0)$. Interestingly, the classification error for the stochastic classifier $h_s$ takes a simple form, which equals the quadratic entropy, i.e., the asymptotic error of a nearest neighbor classifier \citep{cover1967nearest,devroye1996probabilistic}.

\begin{theorem} \label{hsthm}
Suppose there is a unique $(1-\mu)$ quantile so that $\Pb( \gamma > q)=\mu$. Then for  the quantile-threshold classifier $h_q$ defined in \eqref{eq:quantile} we have
$$ \mathcal{E}_q \equiv \mathcal{E}(h_q) = 2 \E[ \gamma(\bX) \one\{\gamma(\bX) \leq q\}] \leq \mathcal{E}(h) $$
for any strength-calibrated $h: \mathcal{X} \mapsto \{0,1\}$ with $\E(h)=\mu$.

Further, the only classifier that is both strength-calibrated and distribution-matched is the stochastic classifier $h_s$ defined in \eqref{eq:h2}. Its error is given by
$$ \mathcal{E}_s  \equiv \mathcal{E}(h_s) = 2 \E\{ \gamma(\bX) - \gamma(\bX)^2 \} . $$
\end{theorem}

\noindent We prove Theorem \ref{hsthm} and all other results in the Appendix. Since $\mathcal{E}_q \leq \mathcal{E}_s$ and $\mathcal{E}_s$ equals the asymptotic nearest-neighbor error, we can transport results from the standard classification setting accordingly. The following theorem from \citet{cover1967nearest,devroye1996probabilistic} shows how these errors yield bounds on the optimal error $\mathcal{E}(h_0)$, and indicates how much worse they can be compared to  $\mathcal{E}(h_0)$. 

\begin{proposition}[\citet{cover1967nearest,devroye1996probabilistic}] \label{thmdgl}
Suppose there is a unique $(1-\mu)$-quantile $q$ such that $\Pb( \gamma > q)=\mu$. Then the optimal classification error $\mathcal{E}(h_0)$ is bounded as
$$ \frac{1}{2} \left( 1 - \sqrt{1 - 2 \mathcal{E}_s} \right) \leq \frac{1}{2} \left( 1 - \sqrt{1 - 2 \mathcal{E}_q} \right) \leq \mathcal{E}(h_0) \leq \mathcal{E}_q \leq \mathcal{E}_s . $$
We further have the upper bound $\mathcal{E}_q \leq \mathcal{E}_s \leq \mathcal{E}(h_0) \{ 1- \mathcal{E}(h_0)\} \leq 2 \mathcal{E}(h_0)$. 
\end{proposition}
\noindent Proposition \ref{thmdgl} follows from our Theorem \ref{hsthm} together with Theorem 3.1 of \citet{devroye1996probabilistic}, and shows that the errors of the stochastic and quantile classifiers can be quite informative about the optimal error $\mathcal{E}(h_0)$ of unconstrained classifiers. For example, if compliance status can be correctly predicted for 75\% of the population with either classifier (e.g., $\mathcal{E}_s= 0.25$) then the optimal classifier can have no better than  86\% accuracy.  Theorem~\ref{hsthm} further indicates that the errors $\mathcal{E}_q$ and $\mathcal{E}_s$ can never be worse than twice that of the best unconstrained classifier, which is particularly informative when $\mathcal{E}_q$ or $\mathcal{E}_s$ are not too large. 

\subsection{Estimation} 

The simplest way to estimate the proposed classification rules is via plug-in estimators. For example, the plug-in estimator of the Bayes decision function $h_0$ is given by
\begin{equation}
\widehat{h}_0(\bx) = \one\{ \widehat\gamma(\bx) > 1/2 \} . \label{eq:excesspi}
\end{equation}
Analogs of this estimator have been studied widely in the classification literature \citep{devroye1996probabilistic, audibert2007fast}. However, the form of the Bayes classifier $h_0$, in our setting, brings some additional complications relative to the standard classification setting, since  $\gamma(\bx) = \lambda_1(\bx)-\lambda_0(\bx)$ is a difference in regression functions. For example, the minimax convergence rate for estimating $\gamma$ can depend  not only on the smoothness of $\gamma$, but also on the smoothness of $\lambda_z$ and $\pi$. This is an open problem and beyond the scope of this paper; nonetheless we can still relate the error of $\widehat{h}_0$ to that of $\widehat\gamma$, as in standard classification problems. Specifically, as in Theorem 2.2 of \citet{devroye1996probabilistic} we have
$$ \mathcal{E}(\widehat{h}) - \mathcal{E}(h_0) \leq 2 \| \widehat\gamma - \gamma \| $$
where here and throughout we let $\| f \|^2 = \Pb(f^2) =  \int f(\bo)^2 \ d\Pb(\bo)$ denote the squared $L_2(\Pb)$ norm (in fact the above also holds replacing the $L_2$ with the $L_1$ norm). This shows that consistent estimation of the compliance score $\gamma$ is enough to yield a consistent plug-in estimator of the rule $h_0$, in terms of classification error. 

A plug-in estimator for the quantile rule $h_q$ is given by
\begin{equation} \label{eq:hqhat}
\widehat{h}_q(\bx) = \one\{ \widehat\gamma(\bx) > \widehat{q} \} 
\end{equation}
where $\widehat{q}$ is an estimate of the $(1-\mu)$ quantile of $\gamma$, i.e. an estimate of $q$ for which 
$\Pb(\gamma \leq q)=1-\mu$. For example one could use $\widehat{q}=\widehat{F}^{-1}(1-\widehat\mu)$, for initial estimators $\widehat{F}$ and $\widehat\mu$  of the distribution function and mean of the compliance score, respectively. In the next subsection we will detail an efficient estimator of $\mu$, which is doubly robust and can attain the minimax root-n convergence rate even if $(\widehat\pi,\widehat\lambda_z)$ converge at slower nonparametric rates. Finally a plug-in estimator of the stochastic classification rule is given by 
\begin{equation} \label{eq:hshat}
\widehat{h}_s(\bx) = \one\{\widehat\gamma(\bx) > U\} 
\end{equation}
for $U \sim \text{Unif}(0,1)$, so that $\widehat{h}_s \sim \text{Bernoulli}(\widehat\gamma)$.  We note that, although natural, the plug-in classifiers described above are not necessarily exactly strength-calibrated or distribution-matched when estimated from a finite-sample.
For the plug-in estimators in \eqref{eq:hqhat} and \eqref{eq:hshat} the next result gives a bound, relating excess classification error to error of the estimated compliance score (and quantile estimation error for $\widehat{h}_q$). For the quantile classifier we require a margin condition \citep{audibert2007fast},
which controls the behavior of $\gamma$ around the threshold $q$. Formally, we have the following condition:
\begin{assumption}[Margin Condition]
For some $\alpha > 0$ and for all $t$ we have that,
\begin{align}
\label{eqn:marginsiva}
\Pb(| \gamma - q| \leq t) \lesssim t^\alpha.
\end{align}
\end{assumption} 
\noindent The margin condition requires that there are not too many compliance scores near the quantile $q$. This is essentially equivalent to the margin condition used in standard classification problems \citep{audibert2007fast}, optimal treatment regime settings \citep{van2014optimal,luedtke2016statistical}, as well as other problems involving estimation of non-smooth covariate-adjusted bounds \citep{kennedy2017survivor}. 

Overall, the following result shows that plug-in classifiers using accurate nuisance estimates have small excess error. 
\begin{theorem} \label{hqesterr}
Let $\widehat{h}_q$ and $\widehat{h}_s$ be the plug-in classifiers defined in \eqref{eq:hqhat} and \eqref{eq:hshat}. Then for $\widehat{h}_s$
$$ | \mathcal{E}(\widehat{h}_s) - \mathcal{E}_s | \leq \left( \sqrt{1-2\mathcal{E}_s} \right) \| \widehat\gamma - \gamma \|.$$
Furthermore, under the margin condition, for $\widehat{h}_q$  we have that,
$$ | \mathcal{E}(\widehat{h}_q) - \mathcal{E}_q | \lesssim \Big( \|\widehat\gamma - \gamma \|_\infty + | \widehat{q}-q| \Big)^\alpha.$$
\end{theorem}

\begin{remark}
From a theoretical standpoint, we might consider if the margin assumption may be eliminated in the analysis of the plug-in quantile classifier. In Appendix~\ref{app:nomargin}, we show that if we can obtain reasonable bounds on the errors 
$ \| \widehat\gamma - \gamma \|_\infty$ and $|\widehat{q}-q|$, a slight modification of the plug-in quantile classifier in~\eqref{eq:hqhat} achieves a similar guarantee without the margin assumption.
\end{remark}
 The next result shows a further unique property of $\widehat{h}_s$, which is that it can be used to  estimate complier characteristics of the form $\theta=\E\{ f(\bX) \mid C=1\}$, by simply computing corresponding averages in the group of predicted compliers with $\widehat{h}_s=1$. For example one might be interested in, for a given variable $X_j$, the complier-specific mean $f(\bX)=X_j$ or distribution function $f(\bX)=\one(X_j \leq t)$. The proposed estimator is then given by
\begin{equation} \label{eq:thetahat}
\widehat\theta = \Pn\{ f(\bX) \mid \widehat{h}_s(\bX)=1 \} = \frac{ \Pn\{ f(\bX) \widehat{h}_s(\bX) \} } { \Pn\{ \widehat{h}_s(\bX) \} } 
\end{equation}
where $\Pn$ denotes the empirical measure so that sample averages can be written as $\Pn\{ f(\bO)\} = \frac{1}{n} \sum_{i=1}^n f(\bO_i)$. For simplicity we suppose $\widehat\gamma$ is fit in a separate independent sample; this will be discussed in more detail after stating the result. 

\begin{theorem} \label{thetahat}
Assume that $f$ is bounded,  
then for the estimator $\widehat\theta$ defined in \eqref{eq:thetahat} we have that
$$ |\widehat\theta - \theta| = O_\Pb\left( \frac{1}{\sqrt{n}} + \| \widehat\gamma - \gamma \| \right), $$
whenever $\widehat\gamma$ is constructed from a separate independent sample. 
\end{theorem}

Theorem \ref{thetahat} shows that $\widehat\theta$ is consistent as long as $\widehat\gamma$ is, and that the convergence rate is of the same order as a typical plug-in estimator. This gives an alternative to the weighting approach of \citet{abadie2003semiparametric}. Our approach only requires computing usual statistics among predicted compliers. In general, however, this approach will not be fully efficient, for two reasons. The first is that $\widehat\theta$ is a plug-in estimator, not specially targeted to estimate $\theta$ well (partly evidenced by the first-order bias term $\| \widehat\gamma-\gamma \|$ in its convergence rate). We conjecture that $\widehat\theta$ might be able to attain full nonparametric efficiency under strong smoothness assumptions and for particular $\widehat\gamma$ estimators (e.g., kernel regression with undersmoothing). However, a more flexible approach would be to estimate $\theta$ with an appropriate doubly robust influence function-based estimator. The other reason the estimator $\widehat\theta$ is not fully efficient is because it uses only a single sample split, however this can be remedied by swapping samples and averaging; we formally include this approach in our subsequent proposed estimators of effect bounds and sharpness. Despite disadvantages with respect to efficiency, the proposed plug-in estimator of $\theta$ might be favored in some settings for its simplicity.

\section{Bounding Effects in Identifiable Subgroups}
\label{sec:bounds}

In this section we consider the second feature of so-called sharp instruments: obtaining tighter bounds on effects in identifiable subgroups, i.e., subgroups defined not by principal strata (e.g., compliers) but by observed covariates. In the toy example from Section~\ref{sec:intro} we saw a case where the local effect actually reduced to such a subgroup effect (among those with $X>0.5244\ldots$). This raises the question of when this can occur, and if it cannot, how to quantify the extent to which it can nearly occur. We derive bounds on effects in any identifiable subgroup and derive the corresponding bound length, and characterize the optimal subgroup that minimizes bound length, among all subgroups of a given size. Finally, we propose efficient nonparametric bound estimators, and describe their asymptotic properties.

\subsection{Bounds \& Bound Length}

Define the treatment effect in an identifiable subgroup $\{ \bx : g(\bx)=1\}$ corresponding to an arbitrary measurable subgroup indicator $g: \mathcal{X} \mapsto \{0,1\}$ as
$$ \beta(g) = \E(Y^{a=1} - Y^{a=0} \mid g=1). $$
Our first result gives bounds on this effect under the instrumental variable assumptions, for any given $g$. Before stating our result, let us first introduce some notation. Define,
\begin{align}
\beta_j(g) &= \E \Big\{ \E(V_{j,1} \mid \bX,Z=1) - \E(V_{j,0} \mid \bX,Z=0)  \Bigm| g=1 \Big\} 
\end{align}
for $j \in \{l,u\}$ where
\begin{align}
\label{eqn:siva}
V_{u,1} &= YA + 1-A , \ \ \
& V_{u,0} &= Y(1-A) , \\
V_{l,1} &= YA , 
& V_{l,0} &= Y(1-A) + A. 
\end{align}
With these definitions in place we have the following result.
\begin{theorem} \label{thm:bounds}
Under Assumptions 1--5, and if $\Pb(Y \in [0,1])=1$, the effect $\beta(g)$ in the identifiable subgroup defined by $g: \mathcal{X} \mapsto \{0,1\}$ is bounded as
\begin{align*}
\beta_l(g) \leq \beta(g) \leq \beta_u(g).
\end{align*}
\end{theorem}

Theorem \ref{thm:bounds} generalizes the results of \citet{robins1989analysis, manski1990nonparametric, balke1997bounds} to allow for covariate adjustment and conditional effects; these previous bounds are recovered by taking $\bX=\emptyset$ and $g=1$ with probability one. The logic used in the proof of Theorem \ref{thm:bounds} follows that of this earlier work. Specifically, as shown above, Assumptions 1--5 allow one to express $\beta(g)$ in terms of observed data quantities and two non-identified terms of the form $\E(Y^{a=t} \mid A^{z=t}=1-t, g=1)$ for $t \in \{0,1\}$; bounds are obtained by replacing these latter quantities with their most extreme values of 0 and 1. Note the condition that $Y \in [0,1]$ is immaterial as long as $Y$ is bounded in some finite range $[y_{\min},y_{\max}]$, since then one can work with $Y^*=(Y-y_{\min})/(y_{\max}-y_{\min}) \in [0,1]$ and transform back.

An important consequence of Theorem \ref{thm:bounds} for our work is in the length of the corresponding bounds, which  provides a basis for quantifying near-identification of effects $\beta(g)$ in identifiable subgroups. This length is given in the following corollary.

\begin{corollary} \label{cor:bndlen}
The length of the bounds in Theorem \ref{thm:bounds}, for any subgroup $h$, is 
$$ \ell(g) \equiv \beta_u(g) - \beta_l(g) = \E\{ 1- \gamma(\bX) \mid g=1 \} . $$
\end{corollary}
\noindent Importantly, under Assumptions 1--5 we have 
$$ \ell(g) =\Pb(C=0 \mid g=1) , $$ 
so the bound length is also interpretable as the proportion of non-compliers in the subgroup $\{ \bx : g(\bx)=1\}$. This fact was noted previously for marginal effects (i.e., when $g=1$ with probability one) by \citet{balke1997bounds}, for example. 
It implies that the bounds on the subgroup effect $\beta(g)$ are strictly narrower than those on the average effect $\E(Y^{a=1}-Y^{a=0})$ whenever $\Pb(C=0 \mid g=1)<\Pb(C=0)$, i.e., whenever the proportion of non-compliers in the subgroup is less than the proportion overall.  
In Section~\ref{sec:sharp} we frame this condition in a  different way that shows how it is intimately related to our proposed notion of sharpness. 

Corollary \ref{cor:bndlen} further suggests exploring subgroups that minimize bound length. Among all possible subgroups, the one minimizing bound length is simply that which picks the subject(s) with the maximum compliance score, i.e., $\argmin_h \ell(g) = \one(\gamma=\gamma_{\max})$ for $\gamma_{\max} = \sup_{\bx \in \mathcal{X}} \gamma(\bx)$. However, in general this subgroup will have negligible size (unless there is a non-trivial point mass at $\gamma_{\max}$), leading to estimates with necessarily high finite-sample error. This is similar to the potential disadvantages of the optimal classification rule $h_0$ discussed in Section \ref{sec:classprop}. Therefore, as discussed there, it may be preferable to only consider subgroups of a particular minimum size. 
We let 
\begin{align*}
\mathcal{G}(t)=\{ g : \Pb(g=1)=t\}
\end{align*} 
denote the set of all subgroups of a given size $t$, and we assume there exists a unique quantile $\xi$ such that $\Pb(\gamma > \xi)=t$. 
The following result gives the form of optimal subgroups of a given size. 
\begin{proposition} \label{prop:hopt}
Let $F(t) = \Pb(\gamma \leq t)$ denote the distribution function of the compliance score, then 
the subgroup that minimizes bound length among all those of size at least $t$ is given by
$$ \argmin_{g \in \mathcal{G}(t)} \ \ell(g) = \one\{ \gamma(\bX) > F^{-1}(1-t) \}.$$
\end{proposition}

\noindent Proposition \ref{prop:hopt} shows that, among all subgroups of size $t$, the subgroup that yields the tightest bounds is simply the group with the $100t\%$ highest compliance scores. This is perhaps expected given the form of the bounds from Corollary \ref{cor:bndlen}, and their interpretation as a proportion of non-compliers. Note also that, once we restrict to subgroups of a given size $t$, the minimizers of complier classification error and bound length are the same, i.e.,
$$  \argmin_{g \in \mathcal{G}(t)} \ \mathcal{E}(g) = \argmin_{g \in \mathcal{G}(t)} \ \ell(g) = \one\{ \gamma(\bX) > F^{-1}(1-t) \} . $$
Therefore, for subgroups of a given size, the problems of finding the classifier with best error and the subgroup with tightest bounds are equivalent, both leading to a version of the quantile classifier $h_q$ from Section \ref{sec:classprop}.

This suggests targeting novel subgroup effects of the form
$$ \E\{Y^{a=1} - Y^{a=0} \mid \gamma(\bX) > F^{-1}(1-t) \} . $$
These effects are similar in spirit to those proposed by \citet{follmann2000effect, joffe2003compliance}, which are also conditional on the compliance score, but these prior works 
use parametric models and do not use quantiles. Thus, our proposed effects can be viewed as a nonparametric generalization.

\subsection{Estimation \& Inference}
\label{sec:infl}

Now we turn to estimation and inference for bounds on $\beta(g)$. We focus in particular on $\beta(h_q)$, i.e., the effect among the $100\mu\%$ of the population with the highest compliance scores. Our bound estimators (and sharpness estimators presented in the next section) are built from the efficient influence function, and use sample splitting. These tools are used to combat bias from nonparametric estimation of nuisance functions (e.g., the compliance score $\gamma$), and to allow arbitrary complex and flexible nuisance estimators to be used.

Influence functions are a central element of nonparametric efficiency theory. We refer to \citet{bickel1993efficient}, \citet{van2002semiparametric}, \citet{van2003unified}, \citet{tsiatis2006semiparametric} and others for more detailed information, and so just give some brief description here. The efficient influence function is important because its variance yields a benchmark for nonparametric efficiency, and because it can be used to construct estimators that are in some cases minimax optimal and efficient in nonparametric models. Such estimators are typically doubly robust or have general second-order bias, and so can attain parametric rates of convergence, even in high-dimensional settings where nuisance functions are estimated at slower rates via flexible nonparametric methods. Mathematically, the efficient influence function corresponds to the score function in a one-dimensional submodel that is least favorable, in the sense of having minimal Fisher information for the parameter of interest, across all submodels. We refer to the earlier references for more details. 

To simplify notation in this section, for any random variable $T$ we let $$ \varphi_z(T; \boldsymbol\eta) =  \frac{\one(Z=z)}{\pi_z(\bX)} \Big\{ T - \E(T \mid \bX, Z = z) \Big\} + \E(T \mid \bX, Z=z) $$ denote the uncentered efficient influence function for the parameter $\E\{ \E(T \mid \bX, Z=z)\}$, where $\boldsymbol\eta=\{\pi_z(\bX),\E(T \mid \bX,Z=z)\}$ denotes the relevant nuisance functions. We use $\boldsymbol\eta$ for nuisance functions generally, though the actual functions depend on the choice of $T$. In particular we let 
\begin{align*}
\nu_{j,z}(\bX)=\E(V_{j,1} \mid \bX,Z=1)-\E(V_{j,0} \mid \bX,Z=0),
\end{align*} 
and let $\widehat\nu_{j,z}$ denote an estimate of $\nu_{j,z}$, for variables $V_{j,z}$ defined as in~\eqref{eqn:siva}.

Following \citet{robins2008higher, zheng2010asymptotic, chernozhukov2016double}, we propose to use sample splitting to allow for arbitrarily complex nuisance estimators $\boldsymbol{\widehat\eta}$ and avoid empirical process conditions, by constructing the estimated $\boldsymbol\eta$ values for each subject using data from only other subjects. Specifically we split the data into $K$ disjoint groups by drawing variables $(B_1,\ldots,B_n)$ independent of the data, with $B_i=b$ indicating that subject $i$ was split into group $b$. For example each $B_i$ could be drawn uniformly from $\{1,\ldots,K\}$, or to ensure equally sized groups $(B_1,\ldots,B_n)$ could be drawn uniformly from the set of permutations of sequences containing $n/K$ repetitions of each value of $b \in \{1,\ldots,K\}$. In our analysis we focus on the former setting, and treat $K$ as a fixed constant.  We first estimate the strength of the instrument by the weighted average of corresponding estimators across groups
\begin{align*} 
\widehat\mu =  \sum_{b=1}^K \left\{ \frac{1}{n} \sum_{i=1}^n \one(B_i=b) \right\} \Pn^b\{ \phi_\mu(\bO; \boldsymbol{\widehat\eta}_{\text{-}b}) \} = \Pn\{ \phi_\mu(\bO; \boldsymbol{\widehat\eta}_{\text{-}B}) \}
\end{align*}
where $\Pn^b$ denotes the sub-empirical distribution over the units $\{ i: B_i=b\}$ in group $b$, i.e., $\Pn^b\{ f(\bO)\} = \sum_{i=1}^n f(\bO_i) \one(B_i=b) / \sum_{i=1}^n \one(B_i=b)$, the function 
\begin{align*}
\phi_\mu(\bO;\boldsymbol\eta)=\varphi_1(A;\boldsymbol\eta)-\varphi_0(A;\boldsymbol\eta),
\end{align*} is the uncentered influence function for $\mu=\E(\gamma)$, and here $\boldsymbol{\widehat\eta}_{\text{-}b}$ denotes estimators of $\boldsymbol{\eta}=(\pi_z,\lambda_z)$ constructed using only those units with $B_i \neq b$. 
Then for $j \in \{l,u\}$ we propose estimating $\beta_j(h_q)$ with  $\widehat\beta_j(\widehat{h}_q)$, where
$$ \widehat\beta_j(\widehat{h}_q) = \Pn\Big[ \Big\{ \varphi_1(V_{j,1};\boldsymbol{\widehat\eta}_{\text{-}B}) - \varphi_0(V_{j,0}; \boldsymbol{\widehat\eta}_{\text{-}B}) \Big\} \widehat{h}_{q,\text{-}B}(\bX) \Big] / \Pn\{ \widehat{h}_{q,\text{-}B}(\bX)\} $$
for  $\widehat{h}_{q,\text{-}b}=\one(\widehat\gamma_{\text{-}b} > \widehat{q}_{\text{-}b})$ and $\widehat{q}_{\text{-}b}$ the $(1-\widehat\mu)$ quantile of $\widehat\gamma$ solving $\Pn^b\{ \one( \widehat\gamma_{\text{-}b} > \widehat{q}_{\text{-}b})\} =\widehat\mu$ (at least up to $o_\Pb(1/\sqrt{n})$ error).

Before stating our next result, we define the remainder terms that appear in our result:
\begin{align}
R_{1,n} &=\| \widehat\pi_1 - \pi_1 \| \left(  \max_z  \| \widehat\lambda_z - \lambda_z \|  +  \max_z \| \widehat\nu_{j,z} - \nu_{j,z} \| \right) \\
R_{2,n} &=\left( \| \widehat\gamma - \gamma \|_\infty  + |\widehat{q}-q| \right)^\alpha,\label{eqn:remone}
\end{align}
where $\alpha > 0$ is the margin exponent in~\eqref{eqn:marginsiva}.
The next theorem gives the rate of convergence for our proposed estimator, as well as nonparametric conditions under which it is asymptotically normal and efficient. 

\begin{theorem} \label{betahat}
Assume that $\Pb\{\epsilon \leq \widehat\pi_z(\bX) \leq 1-\epsilon\}=1$ for $z=0,1$ and some $\epsilon>0$, and that $\| \widehat\pi_1 - \pi_1 \| + \max_z  \| \widehat\lambda_z - \lambda_z \|  +  \max_z \| \widehat\nu_{j,z} - \nu_{j,z} \| + \Pb( \widehat{h}_q \neq h_q) = o_\Pb(1)$. 
\begin{enumerate}
\item If the margin condition holds for some $\alpha$, then
\begin{align*}
\widehat\beta_j(\widehat{h}_q) - \beta_j(h_q) = O_\Pb \left( \frac{1}{\sqrt{n}} + R_{1,n} + R_{2,n} \right).
\end{align*}
\item If it also holds that $R_{1,n} + R_{2,n}=o_\Pb(1/\sqrt{n})$ then 
$$ \sqrt{n}\left\{ \widehat\beta_j(\widehat{h}_q) - \beta_j(h_q) \right\} \indist N\left(0,  \var\Big[ \{ \varphi_1(V_{j,1}) - \varphi_0(V_{j,0}) \} h_q - \beta_j(h_q) \phi_\mu \Big] / \mu^2 \right) . $$
\end{enumerate}
\end{theorem}

Theorem \ref{betahat} shows that the error in estimating bounds on $\beta(h_q)$ consists of a doubly robust second-order term $R_{1,n}$ that will be small if either $\pi_z$ or $(\lambda_z,\nu_{j,z})$ are estimated accurately, along with a term $R_{2,n}$ that will be small if the compliance score $\gamma$ is estimated accurately, and particularly so depending on a margin condition. 

If the exponent in the margin condition is too small, e.g., $\alpha \leq 1$, then the proposed estimators will not in general be asymptotically normal or even $\sqrt{n}$-consistent, for example if $\widehat\gamma$ and $\widehat{q}$ are estimated nonparametrically at slower than $\sqrt{n}$-rates. In general we expect the margin condition to be weakest when the instrument is sharper, i.e., $\alpha$ is likely larger for sharper instruments, and smaller for more blunt instruments, since then $\gamma$ is more flat and likely puts more mass around the quantile $q$. In Appendix~\ref{sec:margin} we consider some examples, and illustrate for which $\alpha$ values the condition holds. An alternative approach would be to avoid the margin condition by instead targeting a smooth approximation of the non-smooth functional $\beta_j(h_q)$, e.g., in the same spirit as \citet{kennedy2017nonparametric}, or smooth but wider bounds.

Importantly, under the conditions of Theorem \ref{betahat} that ensure asymptotic normality, one can use the approach of \citet{imbens2004confidence} to construct valid confidence intervals for the partially identified effect $\beta(h_q)$. We implement this in the \verb|npcausal| R package.

\section{Summarizing Sharpness}
\label{sec:sharp}

So far we have discussed two primary features that make an instrument sharp: accurate prediction of compliers, and tight bounds on effects in identifiable subgroups. In this section we present a new summary measure of sharpness that captures these two properties, separate and apart from strength. We characterize how this measure is related to the complier classification error and bound length quantities from previous sections, and discuss efficient nonparametric estimation and inference. We suggest our sharpness measure be reported alongside strength in practice.

\subsection{Proposed Measure \& Properties}

To summarize sharpness we use the proportion of variance in the instrument's compliance explained by covariates, specifically that proportion explained by the highest compliance score values; this is equivalent to the correlation between the true and predicted compliance status. Although we view this measure as a (strength-independent) summary of how well one can predict compliance and obtain tight bounds on identifiable subgroup effects, we refer to it as sharpness for simplicity.

\begin{definition} \label{sharp}
The \textit{sharpness} $\psi$ of instrument $Z$ with latent compliance indicator $C$ and compliance score $\gamma$ is defined as 
$$ \psi = \frac{\cov(C,h_q)}{\var(C)} = \text{corr}(C,h_q) $$
where $h_q=\one\{ \gamma(\bX) > F^{-1}(1-\mu)\}$ is the quantile classifier defined in \eqref{eq:quantile}, which selects subjects with the top $100\mu\%$ compliance scores.
\end{definition}

We will now give some motivation and intuition for our proposed sharpness measure. First, as a ratio of covariances, it is easily interpretable as a measure of variance explained. In particular, it represents the proportion of variation in compliance explained by the highest $100\mu\%$ compliance scores (it is in the unit interval when $\gamma$ is continuously distributed). In this sense is can be viewed as a model-free and population version of a classical $R^2$ measure, indicating to what extent compliance can be predicted by covariates (through the compliance score). In fact sharpness is also the slope of a population regression of compliance $C$ on predicted compliance $h_q$. At one extreme, if the highest compliance scores do not predict compliance at all, i.e., $C \ind h_q$ (say if $\gamma \approx 0.5$ so that $C$ is just a coin flip), then  the sharpness measure is zero. Conversely, if compliance is perfectly predictable, i.e., $C=\one(\gamma>q)=h_q$, then  sharpness is one. For the toy example in Figure 1, the sharpness is 0\%, 40\%, and 100\% for instruments 1, 2, and 3, respectively.

One could substitute other classifiers for $h_q$ and re-define sharpness as 
$\psi(h) = \cov(C,h)/\var(C)$ for some other $h: \mathcal{X} \mapsto \{0,1\}$, such as $h_0$ or $h_s$ discussed in Section~\ref{sec:class}. We focus on $h_q$ for three main reasons: first, it is optimal among classifiers with size $\mu$, i.e., 
\begin{align*}
\argmin_{h \in \mathcal{G}(\mu)} \mathcal{E}(h) = \argmin_{h \in \mathcal{G}(\mu)} \ell(h) = \argmax_{h \in \mathcal{G}(\mu)} \psi(h) = h_q.
\end{align*} 
Second, using the classifier $h_q$ yields simple and interpretable relationships between $\mathcal{E}(h_q)$, $\ell(h_q)$, and $\psi$, as will be discussed shortly; and finally, the classifier has an easy interpretation as selecting the highest $100\mu\%$ of compliance scores.

The proposed sharpness measure is further interpretable since, for any classifier $h: \mathcal{X} \mapsto \{0,1\}$, we show in Appendix~\ref{sec:youdenequiv} that
$$ \frac{ \cov(C,h) }{ \var(C) }  = \Pb(h=1 \mid C=1) - \Pb(h=1 \mid C=0) . $$
Thus, in addition to measuring variance explained, sharpness also measures the difference between true positive and false positive rates. In particular, for the quantile classifier we have that 
\begin{align*}
\psi = \Pb(\gamma > F^{-1}(1 - \mu) \mid C=1) - \Pb(\gamma > F^{-1}(1 - \mu) \mid C=0).
\end{align*} This quantity is typically called the Youden index, and is a popular summary measure of classifier performance \citep{fluss2005estimation,schisterman2005optimal}. 

One might question the additional benefits of reporting sharpness $\psi$, beyond just the classification error $\mathcal{E}(h)$ or bound length $\ell(h)$.  One crucial feature of $\psi$ is that, unlike say classification error $\mathcal{E}(h)$, it is formally separate from instrument strength $\mu$, in the sense of variation independence. This means for example that $\psi$ -- unlike $\mathcal{E}(h)$ -- cannot be small solely due to instrument strength (or lack thereof). As an illustrative example, consider an instrument for which $\gamma=0.05$ with probability one. Then the optimal classifier in terms of prediction error is given by $h_0=0$, and this uninteresting rule classifies 95\% of subjects correctly (an impressive error rate). However, this instrument has zero sharpness in the intuitive sense of the motivating example from Section~\ref{sec:intro}, and this fact is not reflected by the classification error. In particular, with respect to both classification error and strength, the instrument with $\gamma=0.05$ is virtually indistinguishable from one with $\gamma=\Phi(-2.7+1.4x)$ for $X \sim N(0,1)$. Both yield approximately 5\% classification error and strength, but in the latter case more information is available: we know that subjects with larger $x$ values are more likely to be compliers; in fact we have $\psi = \cov(C,h_q)/\var(C) \approx 50\%$ for the second instrument, compared to $\psi=0$ for the first. 

More formally, sharpness and strength are truly separate measures in the sense that they are variation independent in the presence of a continuous covariate. In particular, for an instrument with any given strength $\mu \in [\epsilon,1-\epsilon]$ we can construct a congenial compliance score $\gamma$ with any sharpness value $\psi \in [0,1]$; conversely, for an instrument with any given sharpness $\psi \in [0,1]$ we can construct a congenial compliance score with any strength value $\mu \in [\epsilon,1-\epsilon]$. For example, suppose without loss of generality that $X \sim N(0,1)$, which can be satisfied for any continuous covariate $X^*$ with cumulative distribution function $G$ via the transformation $X = \Phi^{-1}\{G(X^*)\}$ for $\Phi$ the $N(0,1)$ distribution function. Then for $\gamma(x) = \Phi(b_0+b_1 x)$ we can always find particular $(b_0, b_1)$ values to satisfy $\E(C) = \mu$ and $\cov(C,h_q) = \psi \mu(1-\mu)$ for any $(\mu,\psi) \in [\epsilon,1-\epsilon]^2$. For the case where $\psi=0$ or $\psi=1$, we can simply take $\gamma=\mu$ and $\gamma=\one\{x >\Phi^{-1}(1-\mu)\}$, respectively. More details are given in Section \ref{sec:varind} of the Appendix, along with a plot to illustrate the bijective relationship between $(\mu,\psi)$ and $(b_0, b_1)$. 

Although compliance status $C$ is not directly observed, sharpness is still identified under usual instrumental variable assumptions, simply because the compliance score is identified.
\begin{proposition} \label{psi_id}
Under Assumptions 1--3 and 5, sharpness is identified as $$\psi = { \E\{ \gamma(\bX) h_q(\bX) - \mu^2\} }/{ \mu(1-\mu) }.$$ 
\end{proposition}
\noindent Proposition \ref{psi_id} follows easily from the definition of sharpness together with the fact that $\E(C \mid \bX)=\gamma(\bX)$, and is of course critical for constructing estimators of sharpness from observed data, which will be presented in the next subsection. 

Having defined, motivated, and identified the sharpness measure $\psi$, we now turn to characterizing its relation to classification error and bound length. The next result shows that, keeping strength fixed, sharper instruments yield more accurate complier classification and tighter bounds on identifiable subgroup effects. 

\begin{theorem} \label{sharprel}
The classification error $\mathcal{E}(h_q)$ and bound length $\ell(h_q)$ can be expressed in terms of strength $\mu$ and sharpness $\psi$ as 
$$ \mathcal{E}(h_q) = 2 \mu(1-\mu) (1-\psi) $$
$$ \ell(h_q) = (1-\mu) (1-\psi) . $$
\end{theorem}

The theorem indicates the precise relationship between complier classification error, bound length, strength, and sharpness for $h_q$. The result follows from the fact that, defining $\psi(h)=\cov(C,h)/\var(C)$ for general classifiers $h$, we have
$$ \mathcal{E}(h) = 2 \mu(1-\mu)\{ 1 - \psi(h) \} + (1 - 2\mu) (\E h - \mu) \ , 
\  \ell(h) = (1-\mu)  \{1 - \mu \psi(h) / \E h \}$$
together with the fact that $\E(h_q)=\mu$. Theorem \ref{sharprel} has several important consequences. First, it shows that strength and sharpness are fundamental aspects of the quality of an instrument, since together they completely determine the best error for classifying compliers and the tightest bounds on identifiable subgroup effects, among all classifiers/subgroups of size $\mu$. It also shows that for fixed strength, sharper instruments yield better complier classification and tighter bounds on identifiable subgroup effects. As expected, perfect complier prediction $\mathcal{E}(h_q)=0$ and point identification $\ell(h_q)=0$ requires perfect sharpness $\psi=1$ (note we must have $\mu \leq 1-\epsilon$ because if $\mu=1$ then $A=Z$, which means the instrument cannot be unconfounded if the treatment is confounded). 

Of more practical relevance, Theorem \ref{sharprel} also shows that non-zero sharpness is an important sufficient condition for better complier prediction and tighter bounds on identifiable subgroup effects. Focusing first on complier prediction, we observe that if 
$\psi>0$ then there exists a classifier that attains better error than the naive strength-calibrated classifier (which simply flips a coin with probability $\mu$). This follows since if $\psi>0$ then $\mathcal{E}(h_q)<2\mu(1-\mu)$, which is the error of the rule $h \sim \text{Bern}(\mu)$. Further, since the classifier $h_q$ attains a better error than the coin flip rule, then $h_0$ does as well, since the error of $h_q$ is a lower bound for the latter. 
Turning our attention to bound lengths, we note that if $\psi>0$ then there exists an identifiable subgroup (of size $\mu$) yielding tighter bounds than those on the average treatment effect. This follows since non-zero sharpness $\psi>0$ implies $\ell(h_q) < 1-\mu$, which is the length of the bounds on the average treatment effect $\E(Y^{a=1}-Y^{a=0})$ as derived for example by \citet{robins1989analysis, manski1990nonparametric, balke1997bounds}. The size of $\psi$ indicates the percent reduction in the length of the bounds, e.g., bounds on the subgroup effect $\beta(h_q)$ are precisely $100\psi\%$ tighter than those on the average treatment effect. The only way tighter bounds could be obtained would be to consider smaller subgroups. 

In summary, the sharpness measure proposed in Definition 1 captures the proportion of variance in an instrument's compliance explained by the highest compliance scores. It is an interpretable and strength-independent reflection of (i) how accurately compliers can be classified and (ii) how tightly effects in identifiable subgroups can be bounded.  We suggest that it be reported alongside strength in instrumental variable analyses; in the next subsection we propose methods for estimation and inference.

\subsection{Estimation \& Inference}

Here we propose an estimator for sharpness $\psi$ that, like estimators from previous sections, uses influence functions to correct bias from nonparametric nuisance estimation, and incorporates sample splitting to avoid empirical process restrictions. We refer back to Section~\ref{sec:infl} for more details and notation. 

Our sharpness estimator relies on the strength estimator $\widehat\mu=\Pn\{ \phi_\mu(\bO; \boldsymbol{\widehat\eta}_{\text{-}B}) \}$ from Section~\ref{sec:bounds}, as well as an estimator $\widehat\xi=\Pn(\widehat\phi_{\xi,\text{-}B})$ of 
\begin{align*}
\xi=\E( \gamma h_q),
\end{align*} where  $\phi_{\xi}=\phi_\mu(\bO;\boldsymbol{\eta}) h_q(\bX)$ and $\widehat\phi_{\xi,\text{-}b}=\phi_\mu(\bO;\boldsymbol{\widehat\eta}_{\text{-}b}) \widehat{h}_{q,\text{-}b}(\bX)$ are the corresponding influence function for $\xi$ and its estimate.  In particular, we estimate sharpness as
$$ \widehat\psi = \frac{ (\widehat\xi - \widehat\mu^2 )}{\widehat\mu(1-\widehat\mu)} $$ 
which appropriately combines influence-function-based estimators of the corresponding terms from Definition 1 (i.e., the numerator is the estimator of the covariance between compliance $C$ and the classifier $h_q$). To concisely state our results we define the remainder terms:
\begin{align*}
R_{1,n} &=\| \widehat\pi_1 - \pi_1 \| \left(  \max_z  \| \widehat\lambda_z - \lambda_z \| \right) \\
R_{2,n} &=\left( \| \widehat\gamma - \gamma \|_\infty  + |\widehat{q}-q| \right)^{1+\alpha},
\end{align*}
where once again $\alpha > 0$ is the margin exponent (see~\eqref{eqn:marginsiva}). We note that in comparison with the remainder in~\eqref{eqn:remone} for the estimation of sub-group effects the remainder here for the estimation of sharpness is of lower order, i.e. we are able to estimate sharpness at much faster rates. 
With these definitions in place, the next theorem gives corresponding convergence rates, as well as conditions under which $\widehat\psi$ is asymptotically normal and efficient. 

\begin{theorem} \label{psihat}
Assume  that $\Pb\{\epsilon \leq \widehat\pi_z(\bX) \leq 1-\epsilon\}=1$ for $z=0,1$ and some $\epsilon>0$, and $\| \widehat\pi_1 - \pi_1 \| + \max_z \| \widehat\lambda_z - \lambda_z \| + \Pb(\widehat{h}_q \neq h_q)=o_\Pb(1).$ 
\begin{enumerate}
\item If the margin condition holds for some $\alpha > 0$ the,
$$  \widehat\psi - \psi = O_\Pb\bigg(  \frac{1}{\sqrt{n}} + R_{1,n} + R_{2,n}  \bigg) . $$
\item If it also holds that $R_{1,n} + R_{2,n} = o_\Pb(1/\sqrt{n})$ then
$$ \sqrt{n} (\widehat\psi - \psi) \indist N\left(0, \var\left[ \frac{\{ \phi_\mu h_q + q(\phi_\mu - h_q)  - \xi \}}{(\mu-\mu^2)}  + \frac{ (2\mu\xi - \xi - \mu^2) }{ (\mu-\mu^2)^2 } (\phi_\mu - \mu) \right]  \right) . $$
\end{enumerate}
\end{theorem}

Theorem \ref{psihat} gives two main results. First, it shows that the proposed sharpness estimator is consistent with convergence rate that is second-order in nuisance estimation errors, under weak conditions (bounded IV propensity scores, consistent nuisance estimators, and the margin condition). This means  $\widehat\psi$ attains faster rates than those of its nuisance estimators, which comes from using influence functions for better bias correction than a general plug-in. We do not require any complexity or empirical process conditions, since we use sample splitting to separate the evaluation and estimation of the influence function. Second, Theorem \ref{psihat} shows that if the second-order nuisance errors converge to zero at a faster than $\sqrt{n}$ rate, the estimator is asymptotically normal, and efficient by virtue of the fact that we are working in a nonparametric model (where the only influence function is the efficient one). This condition on the nuisance estimation is satisfied, for example, if $\alpha=1$ and the nuisance estimators converge at faster than $n^{1/4}$ rates; this can hold under nonparametric smoothness, sparsity, or other structural conditions. 

The  asymptotic variance in the second part of Theorem \ref{psihat} can be easily estimated with its corresponding plug-in, 
from which Wald-type confidence intervals can be constructed. Since such intervals may go outside the unit interval, we give an improved logit-transformed interval in Corollary \ref{logit} in the Appendix (which is implemented in the \verb|npcausal| R package).

\section{Simulations \& Illustration}
\label{sec:sims}
In this section we report the results of various simulations
we performed to illustrate the finite-sample performance of our proposed estimators. We also analyze data from a study of canvassing effects on voter turnout \citep{green2003getting} and study the sharpness of the instrument and explore some of its consequences. 

\subsection{Simulation Study}
To assess finite-sample performance, we considered simulations from the following model:
\begin{equation*}
\begin{gathered}
X \sim N(0,1),  \ C \mid X \sim \text{Bern}(\gamma) \text{ for } \gamma(x) = \Phi(b_0 + b_1 x) , \\
Z \mid X, C \sim \text{Bern}(\pi_1) \text{ for } \pi_1(x) = \expit(x) , \\
A = CZ + (1-C) A^* \text{ for } A^* \mid X,C,Z \sim \text{Bern}(0.5) , \\
Y = AY^{a=1} + (1-A) Y^{a=0} \text{ for } Y^a \mid X,C,Z,A \sim \text{Bern}(0.5 + (a-0.5)\beta) ,
\end{gathered}
\end{equation*}
with $(b_0,b_1)$ chosen to ensure given values $(\mu,\psi)$ of strength $\mu=30\%$ and sharpness as detailed in Appendix~\ref{sec:varind}. This model satisfies Assumptions 1--5 and implies 
\begin{equation} \label{sim:params}
\E(Y^{a=1}-Y^{a=0}) = \E(Y^{a=1}-Y^{a=0} \mid h_q)  = \beta . 
\end{equation}
We used the proposed methods to classify compliers and estimate sharpness and bounds. Nuisance functions were estimated with correctly specified logistic regression models, with $K=2$ sample splits. To assess performance we used empirical error $\Pn(\widehat{h} \neq C)$ for each classifier; length of estimated bounds for parameters \eqref{sim:params} with $\beta=20\%$; and bias, RMSE, and 95\% CI coverage of the sharpness estimator. All code is in Appendix~\ref{sec:code}. 

\begin{table}[h!]
\caption{Simulation results across 500 simulations (all figures are percentages). \label{tab:simtab}}
\begin{center}
\begin{tabular}{l rrr rr rrr}
\toprule
& \multicolumn{3}{c}{Class.\ error} & \multicolumn{2}{c}{Bound length}  & \multicolumn{3}{c}{Sharpness est.\ }  \\
\cmidrule(lr){2-4} \cmidrule(lr){5-6} \cmidrule(lr){7-9}
Setting & $\widehat{h}_0$ & $\widehat{h}_q$ & $\widehat{h}_s$  & ATE & $\beta(h_q)$ & Bias & SE & Cov \\
\midrule
$n=500$: & & & & & & & & \\
\ \ \ $\psi=0.2$ & 30.9 & 36.7 & 39.9 		& 68.8 & 61.2 			& -9.4 & 13.5 & 96.9 		\\
\ \ \ $\psi=0.5$ & 21.2 & 22.1 & 29.2 		& 69.9 & 36.2 			& -1.4 & 13.9 & 98.2 		\\
\ \ \ $\psi=0.8$ & 8.5 & 9.3 & 14.3 		& 70.4 & 13.9 			& 0.1 & 10.3 & 95.8 		\\
$n=1000$: & & & & & & & & \\
\ \ \ $\psi=0.2$ & 30.0 & 35.1 & 39.6 		& 70.1 & 59.7 			& -3.7 & 10.3 & 97.0 		\\
\ \ \ $\psi=0.5$ & 20.6 & 21.4 & 28.9 		& 69.9 & 35.4 			& -0.4 & 8.0 & 95.2 		\\
\ \ \ $\psi=0.8$ & 8.4 & 8.8 & 13.6 		& 70.1 & 14.4 			& -0.4 & 7.0 & 95.0 		\\
$n=5000$: & & & & & & & & \\
\ \ \ $\psi=0.2$ & 29.6 & 33.7 & 39.4 		& 70.0 & 56.4 			& -0.4 & 3.6 & 95.8 		\\
\ \ \ $\psi=0.5$ & 20.5 & 21.0 & 28.1 		& 70.1 & 34.9 			& 0.4 & 3.1 & 95.2 		\\
\ \ \ $\psi=0.8$ & 8.4 & 8.5 & 12.6 		& 70.0 & 14.1 			& -0.1 & 3.1 & 94.6 		\\
\bottomrule
\end{tabular}
\end{center}
\end{table}

The simulations illustrate what our theory predicts. Instruments with the same strength can yield drastically different complier classification error (between 39.9\% to 8.4\% here) and subgroup effect bound lengths (between 13.9\% to 70.4\%) depending on sharpness. Our proposed sharpness estimator has minimal bias decreasing with sample size, and confidence intervals attain nominal coverage (coverage was at least 95\% for all bound estimators).

\subsection{Data Analysis}

Here we illustrate the proposed methods by analyzing data from a study of canvassing effects on voter turnout. \citet{green2003getting} conducted a study of $n$ = 18,933 voters across six cities who were randomly assigned to receive encouragement to vote in local elections or not. Recall we are using an iid assumption; inference without this assumption is an important avenue of future research. Non-compliance arose since some voters who were assigned to receive encouragement could not be contacted. As a result \citet{green2003getting} estimated the complier average effect, where here compliers are those people who would be encouraged only when assigned to be. \citet{aronow2013beyond} argue that the local estimand is of limited interest, since in this study compliance is less an inherent characteristic, and more a feature of the design and could change over time (e.g., multiple contacts could increase compliance). Thus it is of interest to identify compliers based on observed characteristics, so as to better generalize the study results by understanding to which subpopulation the effect corresponds. 

In this study the measured covariates include city indicators (Bridgeport, Columbus, Detroit, Minneapolis, Raleigh, St.\ Paul), party affiliation, prior voting history, age, family size, race, and corresponding missingness indicators. We use our proposed methods to classify compliers, estimate bounds on the average treatment effect as well as the subgroup effect $\beta(h_q)$, and assess sharpness of the instrument (i.e., initial randomization). We used random forests (via the \verb|ranger| R package) to estimate the nuisance functions with $K=2$-fold sample splitting. 

In Figures \ref{fig:gamma} and \ref{fig:class}, we present estimated compliance scores and results from the three proposed complier classification methods, respectively. In both cases we plot the voter's estimated compliance scores against two important covariates: the voter's city and age. The estimated compliance scores ranged from 8\% to 69\% across the voter population. Overall, the results indicate that the set of compliers is very likely to contain people from Raleigh (city 5), across a range of ages, as well as older voters in Detroit (city 3). Relative to the estimated Bayes classifier $\widehat{h}_0$,  the quantile classifier $\widehat{h}_q$ classifies more voters as compliers (30\% versus only 4\%), mostly from Raleigh but also from Bridgeport and St.\ Paul. With the stochastic classifier $\widehat{h}_s$ it is somewhat more difficult to distinguish predicted compliers from the rest, based on city and age; however one can still clearly see overrepresentation in Raleigh and St.\ Paul. The estimated error of the quantile classifier is $2\widehat\mu(1-\widehat\mu)(1-\widehat\psi) = 33.3\%$, which yields bounds $27.2\% \pm 6.1\% = [21.1\%, 33.3\%]$ on the optimal error $\mathcal{E}(h_0)$ from Proposition \ref{thmdgl}.

Our nonparametric doubly robust analysis yielded an estimated local effect very similar to that of \citet{green2003getting} (5.7\%, 95\% confidence interval (CI): 2.5\%-8.9\%). However, we estimate that the instrument in this study was stronger than it was sharp, yielding $\widehat\mu = 30.1\%$ (95\% CI: 29.2\%-31.1\%) but  $\widehat\psi=20.9\%$ (95\% CI: 18.8\%-23.2\%). Figure \ref{fig:bds} shows the estimated local effect, along with bounds on the average treatment effect and subgroup effect $\beta(h_q)$; we used the \citet{imbens2004confidence} approach to construct confidence intervals for the subgroup effects. Although the bounds for the subgroup effect are narrower than for the average treatment effect, they are still relatively wide due to the instrument not being very sharp. In particular, the bounds on $\beta(h_q)$ are 20\% narrower than for the average treatment effect, but still cover zero; the estimated bounds on $\beta(h_q)$ are [-17.1\%, 38.7\%] with 95\% CI [-18.9\%, 41.2\%].

\begin{figure}[h!]
    \centering
    
    \begin{subfigure}[t]{0.51\textwidth}
        \centering
        \includegraphics[width=\linewidth]{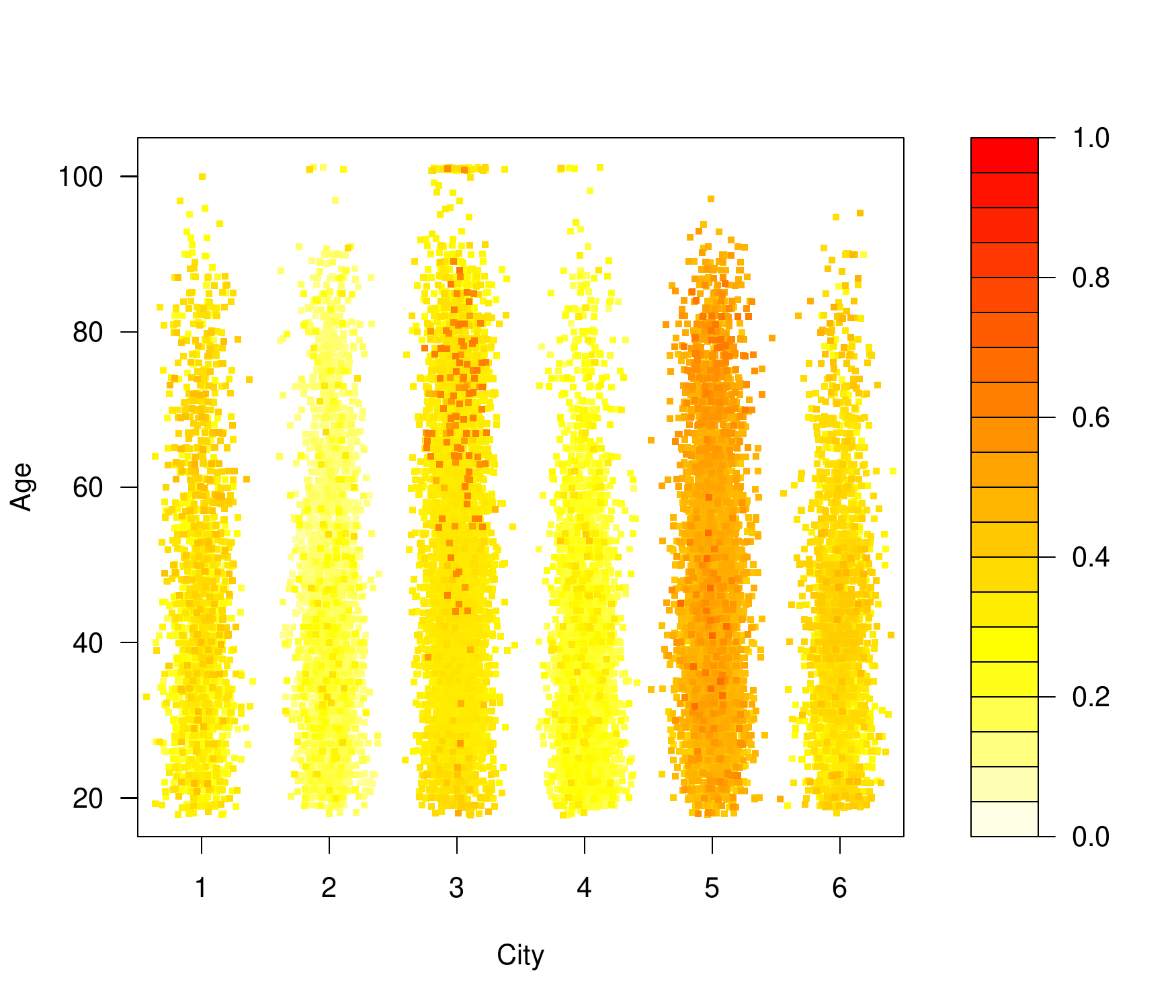} 
        \caption{Estimated compliance scores} \label{fig:gamma}
    \end{subfigure}
    \hfill \hspace*{-.2in}
    \begin{subfigure}[t]{0.445\textwidth}
        \centering
        \includegraphics[width=\linewidth]{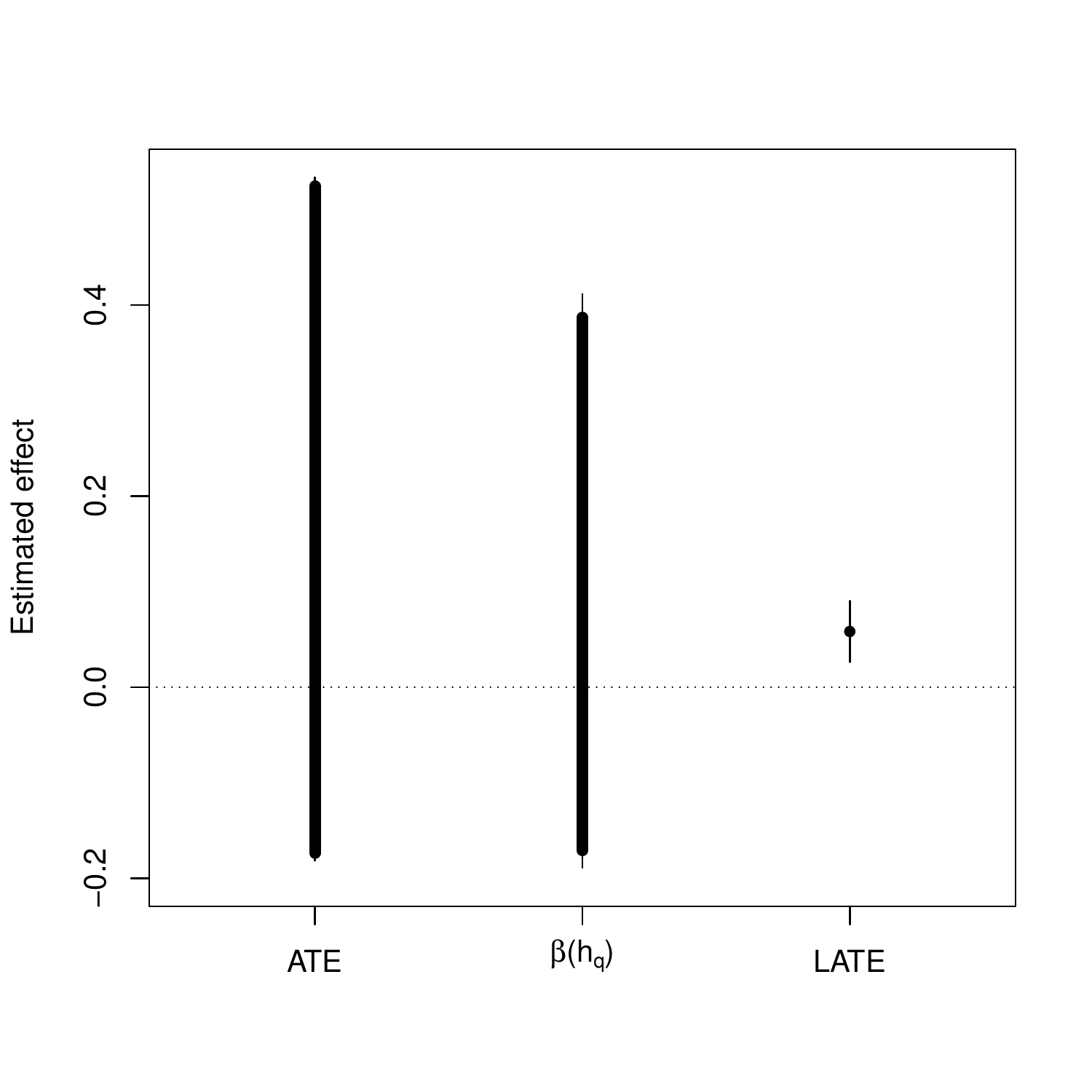} 
        \caption{Estimated effects} \label{fig:bds}
    \end{subfigure}

    \vspace{1cm}
    
    \begin{subfigure}[t]{\textwidth}
    \centering
        \includegraphics[width=\linewidth]{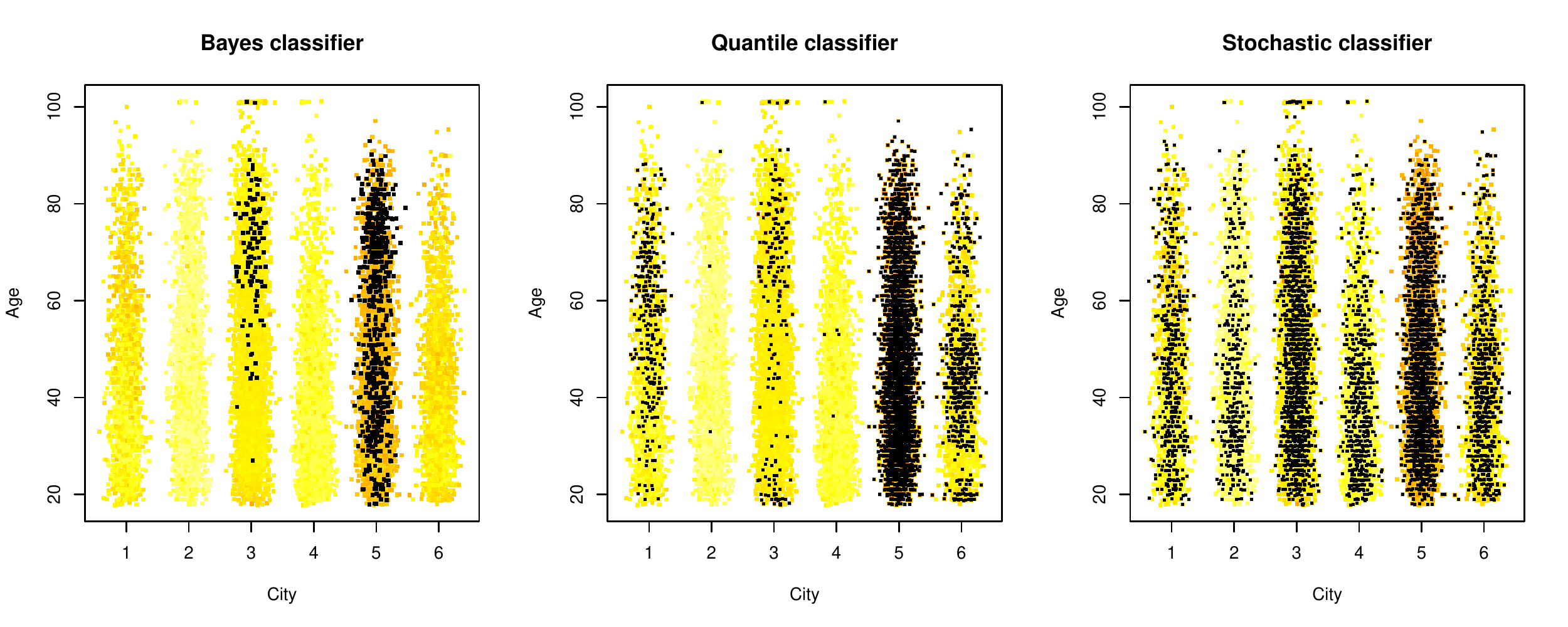} 
        \caption{Complier classification results, predicted compliers are marked in black.} \label{fig:class}
    \end{subfigure}
    
    \caption{Results from analysis of \citet{green2003getting} study of  canvassing effects.}
\end{figure}

\section{Discussion}

In this paper we introduce a new measure of instrument quality, called sharpness, which measures the variation in an instrument's compliance explained by the covariates (in particular, by the compliance scores), and which reflects how well one can predict who compliers are, and how tightly one can bound effects in identifiable subgroups. We propose complier classification rules and characterize their large-sample errors, as well as novel effects in identifiable subgroups defined by subjects with the highest compliance scores. We discuss nonparametric methods for estimating all of these quantities (classification rules, bounds, and sharpness) and give general rates of convergence results, as well as conditions under which the methods are efficient. Finally we have studied the methods via simulation, and applied them in a study of canvassing effects on voter turnout. Implementations of all our methods are publicly available in the \verb|npcausal| R package.

There are several caveats to mention, and ways in which our work could be generalized.  Although we have allowed for complex covariate information, we have focused on the relatively simple setting where both the instrument and treatment are binary. The binary instrument restriction can be removed without changing the estimands and methods too much (although some non-trivial statistical complications could result, as noted for example in \citet{kennedy2018robust}). A multivalued treatment, however, prevents (nonparametric) identification of the compliance score and even local treatment effects; therefore removing this restriction would necessitate a substantially different approach, for example involving estimands that are only partially identified without further assumptions. The same goes for removing the monotonicity restriction, a lack of which also prevents nonparametric identification. Although the binary/monotonic setup we consider here is widely used, it would  be useful in future work to consider analogs of sharpness for different instrumental variable models, such as those of \citet{robins1994correcting, tan2010marginal} that replace monotonicity with effect homogeneity restrictions. It would also be worthwhile to consider violations of the instrumental variable assumptions \citep{balke1997bounds, imbens1997estimating} might affect sharpness. 

In practice, we propose that sharpness should be assessed in instrumental variable studies, alongside strength. Sharp instruments  can help yield more generalizable causal effects (via better prediction of compliers, and tighter bounds on effects in identifiable subgroups), which has been a prominent concern  with standard instrumental variable methods. Given the substantial benefits of sharp instruments, this work also suggests new strategies for data collection and study design. Namely, one should aim to collect data not only on covariates that explain instrument assignment (so as to de-confound the instrument-treatment/outcome relationships for Assumption 3), but also on covariates that predict subjects' compliance behavior. Further, sharpness provides another factor to consider when choosing among instruments, in cases where numerous IVs are available (e.g., in A/B test settings involving many experiments with non-compliance). Importantly, both sharpness and strength can be assessed without outcome data; so if such data collection is costly, one can decide where to collect data on the basis of sharpness and strength.

\section*{References}
\bibliographystyle{abbrvnat}
\bibliography{bibliography}

%


\pagebreak 
\setcounter{page}{1}

\begin{frontmatter}
\title{Supplementary Material to \\ ``Sharp Instruments for Classifying Compliers and Generalizing Causal Effects''}
\runtitle{Sharp Instruments}

\begin{aug}
\author{\fnms{Edward H.} \snm{Kennedy}\ead[label=e1]{edward@stat.cmu.edu}},
\author{\fnms{Sivaraman} \snm{Balakrishnan}\ead[label=e2]{siva@stat.cmu.edu}},
\author{\fnms{Max} \snm{G'Sell}
\ead[label=e3]{mgsell@stat.cmu.edu}}

\runauthor{E.H. Kennedy et al.}

\affiliation{Carnegie Mellon University}

\end{aug}

\end{frontmatter}

\bigskip

\appendix

\section{First-stage F-test example}
\label{app:fstat}
In this section, we demonstrate empirically that the first-stage F-statistic \citep{bound1995problems, staiger1997instrumental, baiocchi2014instrumental} does not capture the sharpness of an instrument. 
Recall that the so-called first-stage F-statistic is the test statistic for a Wald test of the hypothesis that $\gamma(\bx)=0$ in the model
$$ \E(A \mid \bX, Z) = \lambda_0(\bX) + \gamma(\bX) Z . $$
The above formulation is nonparametric. However, in practice these tests are most commonly based on linear models without interactions, so that one assumes $\gamma(\bx)=\beta_1 \in \R$ and $\lambda_0(\bx) = \beta_0^\T \bx$ for $\beta_0 \in \R^p$. Whenever $\gamma(\bx)$ is assumed to be indexed by a finite-dimensional parameter, it is straightforward to also account for interactions. \\

Consider the following example. Suppose $X \sim \text{Unif}(0,1)$, and we have two putative instruments $(Z_1,Z_2)$, both with $Z_j \ind X$. Further suppose the instruments satisfy
$$ \E(A \mid X, Z_1) = 0.2 X + 0.25 Z_1  $$
$$ \E(A \mid X, Z_2) = 0.3 X + Z_2 (0.5 X) $$
so that
$$ \gamma_1(x)=0.25 \ \text{ and } \ \gamma_2(x)=0.5 x . $$
The instruments $Z_1$ and $Z_2$ are equally strong, since for both the proportion of compliers is $25\% = \E(0.5 X)$. However $Z_2$ is clearly a sharper instrument: the sharpness is $0\%$ for $Z_1$ but $\{(1-0.75^2)/4 - 0.25^2\}/(0.25 - 0.25^2) = 25\%$ for $Z_2$. \\

Nonetheless, when $n=1000$, the first-stage F-statistic is \textit{larger} on average for $Z_1$ than for $Z_2$, regardless of whether $\gamma_2(x)$ is modeled correctly. Specifically, based on 1000 simulations, the average F-statistic value is 101.5 for $Z_1$, but is 99.4 for $Z_2$ without testing the interaction (so the model is misspecified), and only 56.3 when also testing the interaction (so the model is correct). This relative ordering persists regardless of sample size. \\

R code for this simulation is given below:

\begin{verbatim}
library(AER); set.seed(1000)
nsim <- 1000; res <- matrix(nrow=nsim,ncol=3)
for (i in 1:nsim){ n <- 1000
  x <- runif(n); z <- rbinom(n,1,.5)
  a <- rbinom(n,1, 0.2*x + 0.25*z)
  fs <- glm(a ~ x+z); fn <- glm(a ~ x)
  res[i,1] <- waldtest(fs,fn, vcov=vcovHC(fs, type="HC0"))$F[2]
  a <- rbinom(n,1, 0.3*x + z*(0 + 0.5*x))
  fs <- glm(a ~ x+z); fn <- glm(a ~ x)
  res[i,2] <- waldtest(fs,fn, vcov=vcovHC(fs, type="HC0"))$F[2]
  fs <- glm(a ~ x*z); fn <- glm(a ~ x)
  res[i,3] <- waldtest(fs,fn, vcov=vcovHC(fs, type="HC0"))$F[2] }
apply(res,2,mean)
\end{verbatim}

\section{Proofs of Propositions \& Technical Lemmas}

\begin{proof}[Proof of Proposition~\ref{errid}]
Since $\E(C \mid \bX)=\gamma$ under Assumptions 1--3 and 5, we have 
$$\mathcal{E}(h) = \E\{ C(1-h) + (1-C)h \} = \E\{ \gamma(1-h) + (1-\gamma)h \}$$
where the last equality follows by iterated expectation.
\end{proof}

\medskip

\begin{proof}[Proof of Proposition~\ref{prop:hopt}]
This follows from the proof of Theorem \ref{hsthm} after noting that
$$ \ell(g) = 1 - \frac{\E(\gamma g)}{\E(g)} = 1 + \frac{\mathcal{E}(g) - \E(\gamma + g)}{2\E(g)} $$
so that minimizing $\ell(g)$ is equivalent to minimizing $\mathcal{E}(g)$ when $\E(g)=t$ is fixed.
\end{proof}

\medskip

In what follows, we give two lemmas used in the main paper. \\

First we provide a result used in the proofs of Theorems \ref{hqesterr} and \ref{betahat}. 
\begin{lemma} \label{ineqlem}
Let $\widehat{f}$ and $f$ take any real values. Then
$$ | \one(\widehat{f}>0) - \one(f>0) | \leq \one( |f| \leq | \widehat{f}-f| ) $$
\end{lemma}
\begin{proof}
This follows since 
\begin{align*}
| \one(\widehat{f}>0) - \one(f>0) | &= \one(\widehat{f},f \ \text{have opposite sign})
\end{align*}
and if $\widehat{f}$ and $f$ have opposite sign then
$$ |\widehat{f}| + |f| = | \widehat{f}-f| $$
which implies $|f| \leq |\widehat{f}-f|$. Therefore whenever $| \one(\widehat{f}>0) - \one(f>0) |=1$ it must also be the case that $\one( |f| \leq | \widehat{f}-f| ) =1$, which yields the result.
\end{proof}

Here we provide a lemma  used to prove Theorems \ref{thetahat}--\ref{betahat}. (Note that the $f$ in Lemma \ref{splitlem} here is unrelated to the $f$ in Theorem \ref{thetahat}.) \\ 

\begin{lemma} \label{splitlem}
Let $\widehat{f}(\bo)$ be a function estimated from a sample $\bO^N=(\bO_{n+1},\ldots,\bO_N)$, and let $\Pn$ denote the empirical measure over $(\bO_1,\ldots,\bO_n)$, which is independent of $\bO^N$. Then 
$$ (\Pn-\Pb) (\widehat{f}-f) = O_\Pb\left( \frac{ \| \widehat{f}-f \| }{\sqrt{n}}  \right) . $$ 
\end{lemma}

\begin{proof}
First note that, conditional on $\bO^N$, the term in question has mean zero since
$$ \E\Big\{ \Pn(\widehat{f}-f) \Bigm| \bO^N \Big\}  = \E(\widehat{f}-f \mid \bO^N) = \Pb(\widehat{f}-f) . $$
The conditional variance is
\begin{align*}
\var\Big\{ (\Pn-\Pb) (\widehat{f}-f) \Bigm| \bO^N \Big\} &=  \var\Big\{ \Pn(\widehat{f}-f) \Bigm| \bO^N \Big\} = \frac{1}{n} \var(\widehat{f}-f \mid  \bO^N ) \leq \|\widehat{f}-f\|^2 /n .
\end{align*}
Therefore using Chebyshev's inequality we have
\begin{align*}
\Pb\left\{ \frac{ | (\Pn-\Pb)(\widehat{f}-f) | }{ \| \widehat{f}-f \| / \sqrt{n} } \geq t \right\} &= \E\left[ \Pb\left\{ \frac{ | (\Pn-\Pb)(\widehat{f}-f) | }{ \| \widehat{f}-f \| / \sqrt{n} } \geq t \Bigm| \bO^N \right\} \right] \leq \frac{1}{t^2} .
\end{align*}
Thus for any $\epsilon>0$ we can pick $t=1/\sqrt{\epsilon}$ so that the probability above is no more than $\epsilon$, which yields the result.
\end{proof}

\section{Proofs of Theorems} \label{sec:thm57}

\begin{proof}[Proof of Theorem~\ref{hsthm}]
Let $h: \mathcal{X} \mapsto \{0,1\}$ be an arbitrary classifier that is strength-calibrated so that $\E(h)=\mu$. If $h \neq h_q=\one\{ \gamma > F^{-1}(1-\mu)\}$ then we can view $h$ as moving some mass away from a region $R_1$ above $F^{-1}(1-\mu)$ to a region below, $R_0$. In particular we can always write
$$ h = \one\{ \gamma > F^{-1}(1-\mu) \} + (f_1 -1) \one(\gamma \in R_1) + f_0 \one(\gamma \in R_0)  $$
for $f_j: \mathcal{X} \mapsto \{0,1\}$, since for any classifier we could just pick $R_j=\{ \gamma : h_q=j\}$ so that $R_0 \cup R_1 = [0,1]$, and let $f_j=h$ for $j \in \{0,1\}$.  \\

Based on the above decomposition of $h$, the strength-calibration constraint that $\E(h)=\mu$ implies 
\begin{equation} \label{eq:constraint}
\E(f_0 \mid \gamma \in R_0) \Pb( \gamma \in R_0) = \E(1-f_1 \mid \gamma \in R_1) \Pb( \gamma \in R_1) . 
\end{equation}
We now define $R_{1,\min}$ and $R_{0,\max}$ as,
\begin{align*}
R_{1,\min} \equiv \inf \{ \gamma: \gamma \in R_1\} \geq \sup \{ \gamma: \gamma \in R_0 \} \equiv R_{0,\max},
\end{align*}
and note that 
\begin{align*}
\frac{\mathcal{E}_q - \mathcal{E}(h)}{2} &=  \E\{\gamma (1- h_q) \} -   \E\{\gamma (1- h) \} =  \E\{ \gamma (h-h_q) \} \\
&=  \E\{ \gamma (f_1 -1) \one(\gamma \in R_1) + \gamma f_0 \one(\gamma \in R_0)  \} \\
&=  \E( \gamma f_0 \mid \gamma \in R_0) \Pb(\gamma \in R_0)  - \E\{ \gamma(1-f_1) \mid \gamma \in R_1\} \Pb(\gamma \in R_1) \\
&\leq  R_{0,\max} \E( f_0 \mid \gamma \in R_0) \Pb(\gamma \in R_0)  - R_{1,\min} \E(1-f_1 \mid \gamma \in R_1) \Pb(\gamma \in R_1) 
\end{align*}
which is non-positive by the constraint \eqref{eq:constraint}. 

The expression for the error $\mathcal{E}_s$ follows since  $\E(h_s \mid \bX)=\gamma$, so that we have $\mathcal{E}(h_s) = \E(\gamma + h_s - 2  \gamma h_s) = 2 \E(\gamma - \gamma^2)$ by iterated expectation. 

To see why $h_s$ is the unique classifier that is strength-calibrated and distribution-matched, note that these properties imply for any $h$ satisfying them that
$$ \E(h \mid \bX) = \gamma(\bX) . $$
This implies that $h$ must be stochastic, i.e., $h: \mathcal{X} \times \mathcal{U} \mapsto \{0,1\}$ for $\mathcal{U}$ the support of some random variable $U$. But the only binary random variable with mean $\gamma$ is Bernoulli, yielding the classifier $h_s$. 
\end{proof}

\medskip

\begin{proof}[Proof of Theorem~\ref{hqesterr}]
Here and throughout we let $\Pb(\widehat{f})= \E\{\widehat{f}(\bO) \mid \bO_1,\ldots,\bO_n\}$ denote expectations over a new observation $\bO$, conditional on the data. For the quantile classifier result, note that
\begin{align*}
|\mathcal{E}(\widehat{h}_q) - \mathcal{E}_q | &= | \Pb\{ ( \gamma + \widehat{h}_q - 2 \gamma \widehat{h}_q) - ( \gamma + {h}_q - 2 \gamma {h}_q) \} | \\
&\leq \Pb\{|1-2\gamma| | \one(\widehat\gamma > \widehat{q}) - \one(\gamma > q) | \} \\
&\leq 2 \Pb\{ (| \gamma - q | + | q -1/2| ) \one( | \gamma - q| \leq | \widehat\gamma - \gamma| + |\widehat{q}-q| ) \} \\
& \leq 2 \Pb\{ ( | \widehat\gamma - \gamma | + |\widehat{q}-q| + |q - 1/2| ) \one( | \gamma - q| \leq | \widehat\gamma - \gamma| + |\widehat{q}-q| ) \}
\end{align*}
where the third line follows by the triangle inequality and Lemma \ref{ineqlem}. Now under the margin assumption $\Pb(|\gamma - q | \leq t) \lesssim t^\alpha$ we have
\begin{align*}
|\mathcal{E}(\widehat{h}_q) - \mathcal{E}_q | &\lesssim ( \| \widehat\gamma - \gamma \|_\infty + |\widehat{q}-q| )^\alpha .
\end{align*}

For the stochastic classifier
\begin{align*}
| \mathcal{E}(\widehat{h}_s) - \mathcal{E}_s | &= | \Pb\{ (\gamma + \widehat{h}_s - 2 \gamma \widehat{h}_s)  - 2 (\gamma-\gamma^2) \} | \\
&= | \Pb\{ (\gamma + \widehat\gamma - 2 \gamma \widehat\gamma)  - 2 (\gamma-\gamma^2) \} | \\
&=  | \Pb\{ (\widehat\gamma-\gamma) (1-2\gamma)  \} | \leq \| 1-2\gamma \| \| \widehat\gamma - \gamma \|
\end{align*}
where the second equality follows by iterated expectation since $\Pb(\widehat{h}_s \mid \bX) = \widehat\gamma$, and the last inequality  by Cauchy-Schwarz. The last line yields the result since 
$$ \| 1 - 2\gamma \| = \sqrt{ \E\{ (1-2\gamma)^2\} } = \sqrt{ 1 - 4 \E(\gamma - \gamma^2)} = \sqrt{1-2\mathcal{E}_s} . $$
\end{proof}

\medskip

\begin{proof}[Proof of Theorem~\ref{thetahat}]
Note that since $\E(h_s \mid \bX) = \gamma$, we have $\Pb(fh_s) = \Pb(f\gamma)$, and therefore $|\widehat\theta-\theta|$ equals
\begin{align*}
\left| \frac{\Pn(f\widehat{h}_s) }{\Pn(\widehat{h}_s) } - \frac{ \Pb(f\gamma)}{\Pb(\gamma) } \right| &= 
\Pb(\gamma)^{-1} \left| \widehat{\theta} \Big\{ \Pb(\gamma) - \Pn(\widehat{h}_s) \Big\} + \left\{\Pn(f\widehat{h}_s) -   \Pb(f\gamma) \right\}\right| \\
&= \Pb(\gamma)^{-1} \left| \widehat{\theta}  \Big\{ (\Pb - \Pn) (\widehat{h}_s) + \Pb (\gamma - \widehat{h}_s) \Big\} \right.\left.+ \Big\{ (\Pn - \Pb) (f\widehat{h}_s) + \Pb (f (\widehat{h}_s - \gamma)) \Big\} \right| \\
&=  \Pb(\gamma)^{-1} \left| \widehat{\theta}  \Big\{ (\Pb - \Pn) (\widehat{h}_s) + \Pb (\gamma - \widehat{\gamma}) \Big\} \right. \left.+ \Big\{ (\Pn - \Pb) (f\widehat{h}_s) + \Pb (f (\widehat{\gamma} - \gamma)) \Big\} \right| 
\end{align*}
By Lemma \ref{splitlem} and the fact that $f$ and $\widehat{h}_s$ are bounded we have that,
\begin{align*}
 (\Pb - \Pn) (\widehat{h}_s) = O_\Pb\left( \frac{1}{ \sqrt{n} }\right)~\text{and}~
 (\Pn - \Pb) (f\widehat{h}_s) = O_\Pb\left( \frac{1}{ \sqrt{n} }\right).
\end{align*}
Furthermore, by the Cauchy-Schwarz inequality we have that,
\begin{align*}
\Pb (\gamma - \widehat{\gamma}) \leq \|\widehat{\gamma} - \gamma\|,~\text{and}~\Pb (f (\widehat{\gamma} - \gamma)) \leq \|f\|_{\infty} \|\widehat{\gamma} - \gamma\|.
\end{align*}
By strong monotonicity $\Pb(\gamma) \geq \epsilon > 0$. Finally we observe that,
\begin{align*}
\widehat{\theta} = \frac{\Pn(f\widehat{h}_s) }{\Pn(\widehat{h}_s) } \leq \|f\|_{\infty},
\end{align*}
which is bounded by assumption, and putting these together we obtain the desired result.
\end{proof}

\medskip

\begin{proof}[Proof of Theorem~\ref{thm:bounds}]
For any subgroup $g$ we have for $\beta(g) = \E(Y^{a=1}-Y^{a=0} \mid g=1)$ and under Assumptions 1--5 that
\begin{align}
\beta(g) &= \E\Big[ (Y^{a=1}-Y^{a=0})  \Big\{ (A^{z=1}-A^{z=0}) + (1-A^{z=1}) + A^{z=0} \Big\} \Bigm| g=1\Big] \nonumber \\
&= \E\Big[ (Y^{a=1}-Y^{a=0}) \Big\{ \one(A^{z=1}>A^{z=0}) \\
& \hspace{.4in} + \one(A^{z=1}=A^{z=0}=0) + \one(A^{z=1}=A^{z=0}=1) \Big\} \Bigm| g=1 \Big] \nonumber \\
&= \E \Big[ \{ \E(Y \mid X,Z=1) - \E(Y \mid X, Z=0) \} + Y^{a=1}(1-A^{z=1}) \nonumber \\
& \hspace{.4in} - \E\{Y(1-A) \mid X,Z=1\} + \E(YA \mid X,Z=0) - Y^{a=0} A^{z=0} \Bigm| g=1 \Big] \nonumber \\
&= \E \Big[ \E(YA \mid X,Z=1) - \E\{Y(1-A) \mid X,Z=0\} \nonumber \\
& \hspace{.4in} + \E(Y^{a=1} \mid A^{z=1}=0, h=1) \E(1-A \mid X,Z=1) \label{eq:unident} \\
& \hspace{.4in} - \E(Y^{a=0} \mid A^{z=0}=1, h=1) \E(A \mid X,Z=0) \Bigm| g=1 \Big] \nonumber . 
\end{align} 
where the first equality follows by definition, the second by monotonicity, the third using standard IV identification and that
$$ Y^{z=1} (1-A^{z=1}) = Y^{z=1,a=A^{z=1}} (1-A^{z=1}) = Y^{z=1,a=0} (1-A^{z=1})  = Y^{a=0} (1-A^{z=1})$$
and similarly
$$ Y^{z=0} A^{z=0} = Y^{z=0,a=A^{z=0}} A^{z=0} = Y^{z=0,a=1} A^{z=0}  = Y^{a=1} A^{z=0} $$
(by consistency and the exclusion restriction), and the fourth by unconfoundedness of $Z$ and iterated expectation. 

Note the terms $\E(Y^a \mid A^{z=a}=1-a, g=1) $ in the last two lines of \eqref{eq:unident} are not identified. If without loss of generality $\Pb(Y \in [0,1])=1$, then plugging in 0 and 1 for these two terms shows that $\beta(g)$ is bounded above by
$$ \beta_u(g)  = \E \Big[ \E(YA + 1-A \mid X,Z=1) - \E\{Y(1-A) \mid X,Z=0\}  \Bigm| g=1 \Big] $$
and below by
$$ \beta_l(g) = \E \Big[ \E(YA \mid X,Z=1) - \E\{Y(1-A) + A \mid X,Z=0\}  \Bigm| g=1 \Big] . $$
\end{proof}

\medskip

\begin{proof}[Proof of Theorem \ref{betahat}]
Here we use similar logic that we develop in more detail in the proof of Theorem~\ref{psihat}.
To ease notation we give results for the case where there are two independent samples of size $n$, one of which is used solely for nuisance estimation. Then the logic from the proof of Theorem \ref{psihat} can be applied to analyze the actual proposed estimator, which randomly splits the sample, estimates nuisance functions on $K-1$ folds and evaluates the estimator on the held-out $K^{th}$ fold, and then swaps and averages. Thus in this proof the estimator is given by
$$ \widehat\beta_j(\widehat{h}_q) = \Pn\left[\{ \varphi_1(V_1;\boldsymbol{\widehat\eta}) - \varphi_0(V_0; \boldsymbol{\widehat\eta}) \} \widehat{h}_q(\bX; \boldsymbol{\widehat\eta}) \right] / \widehat\mu $$
where $(\boldsymbol{\widehat\eta}, \widehat\nu, \widehat\mu)$ are constructed in the independent sample, and we have removed the $j$ subscripts on $\nu$ and variables $V$ to ease notation. 

Let $\varphi_q = \{ \varphi_1(V_1;\boldsymbol{\eta}) - \varphi_0(V_0; \boldsymbol{\eta}) \} {h}_q(\bX; \boldsymbol{\eta})$ with $\widehat\varphi_q = \{ \varphi_1(V_1;\boldsymbol{\widehat\eta}) - \varphi_0(V_0; \boldsymbol{\widehat\eta}) \} \widehat{h}_q(\bX; \boldsymbol{\widehat\eta})$ the corresponding estimated version, and let $\phi_\mu=\varphi_1(A; \boldsymbol{\eta})-\varphi_0(A; \boldsymbol{\eta})$ as in the proof of Theorem \ref{psihat}. Then for $\beta_j=\beta_j(h_q)$ we have
\begin{align*}
\widehat\beta_j(h_q) - \beta_j &= \frac{1}{\Pn\widehat\phi_\mu} \left\{ (\Pn \widehat\varphi_q - \Pb \varphi_q) - \beta_j (\Pn \widehat\phi_\mu - \Pb \phi_\mu) \right\}
\end{align*}
By the results in Theorem \ref{psihat}, if $\widehat\pi_z$ is bounded away from zero, then
$$ \Pn \widehat\phi_\mu - \Pb \phi_\mu = (\Pn - \Pb) \phi_\mu + o_\Pb(1/\sqrt{n}) $$
if $\| \widehat\pi_1 - \pi_1 \| + \max_z \| \widehat\lambda_z - \lambda_z \| = o_\Pb(1)$ and $\| \widehat\pi_1 - \pi_1 \| \left(\max_z \| \widehat\lambda_z - \lambda_z \| \right) =o_\Pb(1/\sqrt{n})$. By the same exact logic as in Theorem \ref{psihat} we also have by Lemma \ref{splitlem} that
$$ \Pn  \widehat\varphi_q - \Pb \varphi_q = (\Pn-\Pb) \varphi_q + \Pb(\widehat\varphi_q - \varphi_q) +  o_\Pb(1/\sqrt{n}) $$
as long as $\| \widehat\varphi_q - \varphi_q \| \lesssim \| (\widehat\varphi_1-\widehat\varphi_0) - (\varphi_1-\varphi_0) \| + \| \widehat{h}_q - h_q \| = o_\Pb(1)$, which follows if
$$ \| \widehat\pi_1 - \pi_1 \| + \max_z \| \widehat\nu_z - \nu_z \| + \Pb( \widehat{h}_q \neq h_q) = o_\Pb(1) . $$
Now we have
\begin{align*}
\Pb(\widehat\varphi_q - \varphi_q) &= \Pb\{ (\widehat\varphi_1-\widehat\varphi_0) \widehat{h}_q - (\varphi_1-\varphi_0) h_q \} \\
&= \Pb\Big[ \{ (\widehat\varphi_1-\widehat\varphi_0) - (\varphi_1-\varphi_0)  \} \widehat{h}_q + (\varphi_1-\varphi_0) (\widehat{h}_q - h_q) \Big] .
\end{align*} 
From Theorem \ref{psihat}, the first term in the last line above will be $o_\Pb(1/\sqrt{n})$ if 
$$ \| \widehat\pi_1 - \pi_1 \| \left(\max_z \| \widehat\nu_z - \nu_z \| \right) =o_\Pb(1/\sqrt{n}) $$
and the second equals
\begin{align*}
\Pb\{ \nu (\widehat{h}_q - h_q) \} &\leq \Pb\{ \nu | \one(\widehat\gamma > \widehat{q}) - \one(\gamma > q) | \} \\
&\lesssim \Pb( | \gamma - {q} | \leq | \widehat\gamma - \gamma | + | \widehat{q} - q | ) \\
&\lesssim ( \| \widehat\gamma - \gamma \|_\infty + | \widehat{q} - q | )^\alpha
\end{align*}
where we used Lemma \ref{ineqlem} with the triangle inequality (and that $\nu$ is bounded) in the second line, and the margin assumption $\Pb( | \gamma - q| \leq t ) \lesssim t^\alpha$ in the third line. This yields the result. 
\end{proof}

\medskip

 \begin{proof}[Proof of Theorem~\ref{sharprel}]
Here we show that
$$ \mathcal{E}(h) = 2 \mu(1-\mu)\{ 1 - \psi(h) \} + (1 - 2\mu) (\E h - \mu) \ , 
\  \ell(h) = (1-\mu)  \{1 - \mu \psi(h) / \E h \}$$
from which Theorem \ref{sharprel} follows by the fact that $\E(h_q)=\mu$. 

First note that since $\psi(h)=\cov(C,h)/\var(C)=\E\{(\gamma - \mu)h\}/(\mu-\mu^2)$ we have
\begin{align*}
2\mu(1&-\mu)\{ 1 - \psi(h) \} + (1-2\mu)(\E h - \mu) \\
&= 2\{ \mu - \mu^2 - \E(\gamma h) + \mu \E(h) \} + (\E h - \mu - 2\mu \E h + 2 \mu^2) \\
&= \mu + \E h - 2 \E(\gamma h) = \mathcal{E}(h) .
\end{align*}

Similarly,
\begin{align*}
(1-\mu) \left\{ 1 - \frac{ \mu}{\E(h)} \psi(h) \right\} &= (1-\mu) \left\{ 1 - \frac{\E(\gamma h ) - \mu \E(h)}{(1-\mu) \E(h) } \right\} \\
&= 1 - \E(\gamma \mid h=1) = \ell(h) .
\end{align*}
\end{proof}

\medskip

\begin{proof}[Proof of Theorem \ref{psihat}]
First consider estimation of $\widehat\mu$. We let $\phi_{\mu}=\varphi_1(A;\boldsymbol{\eta})-\varphi_0(A;\boldsymbol{\eta})$ and $\widehat\phi_{\mu,\text{-}b}=\varphi_1(A;\boldsymbol{\widehat\eta}_{\text{-}b})-\varphi_0(A;\boldsymbol{\widehat\eta}_{\text{-}b})$ to ease notation. Note we can write  
$$ \widehat\mu = \sum_{b=1}^K \Pn\{ \widehat\phi_{\mu,\text{-}b} \ \one(B=b) \} \ \text{ and } \ \mu = \E(\phi_\mu)= \sum_{b=1}^K \Pb\{ \phi_{\mu} \ \one(B=b) \} . $$
Therefore
\begin{equation} \label{eq:muexp}
\widehat\mu - \mu = (\Pn-\Pb) \phi_\mu  + \sum_{b=1}^K \left[ (\Pn-\Pb) \{  ( \widehat\phi_{\mu,\text{-}b} - \phi_\mu ) \one(B=b) \}  + \Pb\{  ( \widehat\phi_{\mu,\text{-}b} - \phi_\mu ) \one(B=b) \} \right] .
\end{equation}
Now note that 
\begin{align*}
\| ( \widehat\phi_{\mu,\text{-}b} - \phi_\mu &) \one(B=b) \| \leq \| \widehat\phi_{\mu,\text{-}b} - \phi_\mu \| \lesssim \| \widehat\phi_{\mu} - \phi_\mu \|   \\
&= \left\| \frac{(2Z-1)(A-\widehat\lambda_Z)}{\pi_Z \widehat\pi_Z} ( \pi_Z - \widehat\pi_Z ) + \frac{(2Z-1)}{\pi_Z} (\lambda_Z - \widehat\lambda_Z) + (\widehat\gamma - \gamma) \right\| \\
&\lesssim \| \widehat\pi_1 - \pi_1 \| + \max_z \| \widehat\lambda_z - \lambda_z \|
\end{align*}
where the last result of the first line follows because $K$ is fixed so $n \lesssim n/K$, and the last line follows since $\widehat\pi_z$ is bounded away from zero. Therefore the first term inside the sum in \eqref{eq:muexp} is $o_\Pb(1/\sqrt{n})$ by Lemma~\ref{ineqlem}, since $\| \widehat\pi_1 - \pi_1 \| + \max_z \| \widehat\lambda_z - \lambda_z \| = o_\Pb(1)$ by assumption. For the second term inside the sum in  \eqref{eq:muexp} we similarly have
\begin{align*}
| \Pb\{  ( \widehat\phi_{\mu,\text{-}b} - \phi_\mu ) \one(B=b) \} | &\lesssim |\Pb( \widehat\phi_{\mu} - \phi_\mu ) | \\
&= | \Pb \{ \widehat\pi_1^{-1} (\pi_1 - \widehat\pi_1) (\lambda_1 - \widehat\lambda_1) + \widehat\pi_0^{-1} (\pi_1 - \widehat\pi_1) (\lambda_0 - \widehat\lambda_0) \} | \\
&\lesssim \| \widehat\pi_1 - \pi_1 \| \left(\max_z \| \widehat\lambda_z - \lambda_z \| \right) .
\end{align*}
Therefore we have 
$$ \widehat\mu - \mu = (\Pn-\Pb) \phi_\mu + O_\Pb\left(  \| \widehat\pi_1 - \pi_1 \| \left( \max_z \| \widehat\lambda_z - \lambda_z \|\right) \right) + o_\Pb(1/\sqrt{n})  . $$

Now let $\phi_{\xi}=\phi_\mu(\bO;\boldsymbol{\eta}) h_q(\bX)$ and $\widehat\phi_{\xi,\text{-}b}=\phi_\mu(\bO;\boldsymbol{\widehat\eta}_{\text{-}b}) \widehat{h}_{q,\text{-}b}(\bX)$, and consider the estimator $\widehat\xi=\Pn(\widehat\phi_{\xi,\text{-}B})$ of $\xi=\E( \gamma h_q)=\Pb(\phi_\xi)$. Therefore we can write $\widehat\xi-\xi$ as
\begin{align} \label{eq:xiexp}
 (\Pn-\Pb) (\widehat\phi_{\xi,\text{-}B} - \phi_\xi) + (\Pn-\Pb) \phi_\xi + \Pb\{ \widehat{h}_{q,\text{-}B}(\widehat\phi_{\mu,\text{-}B} - \phi_\mu) \} + \Pb\{ \phi_\mu(\widehat{h}_{q,\text{-}B} - h_q) \} .
\end{align}
For the first term, noting that $(\Pn-\Pb)(\widehat\phi_{\xi,\text{-}B} - \phi_\xi) = \sum_{b=1}^K (\Pn-\Pb)\{(\widehat\phi_{\xi,\text{-}b} - \phi_\xi)\one(B=b)\}$, we have
\begin{align*}
\| ( \widehat\phi_{\xi,\text{-}b} - \phi_\xi ) \one(B=b) \| &\leq \| \widehat\phi_{\xi,\text{-}b} - \phi_\xi \| \lesssim \| \widehat\phi_{\xi} - \phi_\xi \|   \\
&= \| \widehat{h}_q (\widehat\phi_\mu - \phi_\mu) + \phi_\mu (\widehat{h}_q - h_q) \| \lesssim \| \widehat\phi_\mu - \phi_\mu \| + \| \widehat{h}_q - h_q \| \\
&\lesssim \| \widehat\pi_1 - \pi_1 \| + \max_z \| \widehat\lambda_z - \lambda_z \| + \Pb(\widehat{h}_q \neq h_q)
\end{align*}
so the first term in \eqref{eq:xiexp} is $o_\Pb(1/\sqrt{n})$ by Lemma~\ref{ineqlem}. Similarly, for the third term
\begin{align*}
\Pb\{ \widehat{h}_{q,\text{-}B} (\widehat\phi_{\mu,\text{-}B}  - \phi_\mu) \} & = \sum_{b=1}^K \Pb\{ \widehat{h}_{q,\text{-}b} (\widehat\phi_{\mu,\text{-}b}  - \phi_\mu) \} \Pb(B=b) \lesssim \Pb\{\widehat{h}_q (\widehat\phi_\mu - \phi_\mu)\} \\
&= \Pb[ \widehat{h}_q \{ \widehat\pi_1^{-1} (\pi_1 - \widehat\pi_1) (\lambda_1 - \widehat\lambda_1) + \widehat\pi_0^{-1} (\pi_1 - \widehat\pi_1) (\lambda_0 - \widehat\lambda_0) \} ] \\
&\lesssim \| \widehat\pi_1 - \pi_1 \| \left( \max_z \| \widehat\lambda_z - \lambda_z \| \right) . 
\end{align*}
For the fourth term in \eqref{eq:xiexp} we have
\begin{align} \label{eq:xilast}
\Pb\{ \phi_\mu(\widehat{h}_{q,\text{-}B} - h_q) \}  &= \Pb\{ \gamma (\widehat{h}_{q,\text{-}B} - h_q) \} = \Pb\{ (\gamma-q) (\widehat{h}_{q,\text{-}B} - h_q) \} +  q \Pb (\widehat{h}_{q,\text{-}B} - h_q) .
\end{align}
For the first term on the far right side of \eqref{eq:xilast} we have
\begin{align*}
\Pb\{ (\gamma-q) (\widehat{h}_{q,\text{-}B} - h_q) \} &= \sum_{b=1}^K \Pb\{ (\gamma-q) (\widehat{h}_{q,\text{-}b} - h_q) \} \Pb(B=b) \\
& \lesssim \Pb\{ (\gamma-q) (\widehat{h}_q - h_q) \} \leq \Pb( |\gamma-q| | \one(\widehat\gamma > \widehat{q}) - \one(\gamma > q) | ) \\
& \leq \Pb\{ |\gamma-q| \ \one( | \gamma - q| \leq | \widehat\gamma - \gamma | + | \widehat{q} - q | ) \} \\
& \leq \Pb\{ (| \widehat\gamma - \gamma | + | \widehat{q} - q |) \ \one( | \gamma - q| \leq | \widehat\gamma - \gamma | + | \widehat{q} - q | ) \} \\
& \lesssim \left( \| \widehat\gamma - \gamma \|_\infty + |\widehat{q}-q| \right)^{1+\alpha} 
\end{align*}
 where the third line follows by Lemma~\ref{ineqlem} and the triangle inequality, the fourth by the indicator condition, and the fifth by the margin condition. Now note
\begin{align*}
\widehat\mu - \mu &= \Pn( \widehat{h}_{q,\text{-}B}) - \Pb (h_q) \\
&= (\Pn-\Pb)( \widehat{h}_{q,\text{-}B}-h_q ) + (\Pn-\Pb) h_q  + \Pb( \widehat{h}_{q,\text{-}B}- h_q ) \\
&=  (\Pn-\Pb) h_q  + \Pb( \widehat{h}_{q,\text{-}B}- h_q ) + o_\Pb(1/\sqrt{n})
\end{align*}
where the last line follows since $\| \widehat{h}_q - h_q \| = o_\Pb(1)$ by the fact that $\Pb( \widehat{h}_q \neq h_q) = o_\Pb(1)$, together with Lemma \ref{ineqlem}. Therefore rearranging yields for the second term in \eqref{eq:xilast}  that
\begin{align*}
q \Pb (\widehat{h}_{q,\text{-}B} - h_q)  &= q(\widehat\mu - \mu)  - (\Pn-\Pb) h_q + o_\Pb(1/\sqrt{n}) \\
&= q(\Pn-\Pb) (\phi_\mu - h_q)  + O_\Pb\left( \| \widehat\pi_1 - \pi_1 \| \left( \max_z \| \widehat\lambda_z - \lambda_z \|\right) \right) + o_\Pb(1/\sqrt{n}) .
\end{align*}
This logic is similar to that of \citet{luedtke2016optimal}, with an additional term due to the fact that the quantile $\widehat\mu$ is estimated.

Putting this all together gives
\begin{align*}
 \widehat\mu - \mu &= (\Pn-\Pb) \phi_\mu + O_\Pb\left(  \| \widehat\pi_1 - \pi_1 \| \left( \max_z \| \widehat\lambda_z - \lambda_z \|\right) \right) + o_\Pb(1/\sqrt{n}) \\
\widehat\xi - \xi &= (\Pn-\Pb) \{ \phi_\mu h_q + q(\phi_\mu - h_q) \} + O_\Pb\left(  \| \widehat\pi_1 - \pi_1 \| \left( \max_z \| \widehat\lambda_z - \lambda_z \|\right) \right) \\
& \hspace{.4in} + O_\Pb \Big( \left( \| \widehat\gamma - \gamma \|_\infty + |\widehat{q}-q| \right)^{1+\alpha}  \Big)  + o_\Pb(1/\sqrt{n})  . 
\end{align*}

The first result of Theorem \ref{psihat} now follows since
\begin{equation} \label{eq:dm_psi}
\widehat\psi - \psi = \frac{1}{\widehat\mu(1-\widehat\mu) } \Big(\widehat\xi - \xi \Big) + \frac{ \{ \xi (\widehat\mu + \mu) - \xi - \widehat\mu \mu \}}{\widehat\mu(1-\widehat\mu) \mu(1-\mu)}  \Big(\widehat\mu-\mu \Big) 
\end{equation}
implies
\begin{align*}
 \widehat\psi - \psi &= O_\Pb\bigg(  \frac{1}{\sqrt{n}} + \| \widehat\pi_1 - \pi_1 \| \left( \max_z \| \widehat\lambda_z - \lambda_z \|\right) +  \left( \| \widehat\gamma - \gamma \|_\infty + |\widehat{q}-q| \right)^{1+\alpha} \bigg) 
\end{align*}

For the second result, note that if 
$$\| \widehat\pi_1 - \pi_1 \| \left( \max_z \| \widehat\lambda_z - \lambda_z \|\right) +  \left( \| \widehat\gamma - \gamma \|_\infty + |\widehat{q}-q| \right)^{1+\alpha} = o_\Pb(1/\sqrt{n})$$
 then
\begin{align*}
\sqrt{n} \left\{ \begin{pmatrix} \widehat\mu \\ \widehat\xi \end{pmatrix} - \begin{pmatrix} \mu \\ \xi \end{pmatrix} \right\} &= \sqrt{n} (\Pn-\Pb) \begin{pmatrix} \phi_\mu \\ \phi_\mu h_q + q(\phi_\mu - h_q) \end{pmatrix} + o_\Pb(1/\sqrt{n}) \\
&\indist N\left( 0, \cov\begin{pmatrix} \phi_\mu \\ \phi_\mu h_q + q(\phi_\mu - h_q) \end{pmatrix}  \right) . 
\end{align*}
Now by the delta method (or \eqref{eq:dm_psi}) this implies
$$ \widehat\psi - \psi = (\Pn-\Pb)  \left\{ \frac{\phi_\mu h_q + q(\phi_\mu - h_q)  - \xi }{(\mu-\mu^2)}  + \frac{ (2\mu\xi - \xi - \mu^2) }{ (\mu-\mu^2)^2 } (\phi_\mu - \mu) \right\} + o_\Pb(1/\sqrt{n}) $$
which yields the result.
\end{proof}

\section{A Modified Plug-In Quantile Classifier}
\label{app:nomargin}
\noindent In this section we describe a modified plug-in quantile classifier and show that it has small excess classification error
even without the margin condition we assumed previously.
Concretely, we suppose  
that we are given
plug-in estimates, $(\widehat{\gamma}, \widehat{q})$ and scalars $(\kappa_1,\kappa_2)$ such that:
\begin{align*}
\|\widehat{\gamma} - \gamma\|_{\infty} \leq \kappa_1,~~~~~|\widehat{q} - q| \leq \kappa_2.
\end{align*}
Our arguments can easily be modified to the setting when the upper bounds on the errors 
$(\kappa_1,\kappa_2)$ only hold with high-probability but we do not consider  this extension for simplicity.

Consider the plug-in type classifier,
\begin{align}
\label{eq:hqhatmod}
\widehat{h}_q(\bx) = \mathbb{I}(\widehat{\gamma}(\bx) \geq \widehat{q} - (\kappa_1 + \kappa_2) \mathbb{I}(\widehat{\gamma}(\bx) \geq 1/2) + (\kappa_1 + \kappa_2) \mathbb{I}(\widehat{\gamma}(\bx) < 1/2)).
\end{align}
Intuitively, this classifier modifies the quantile classifier in~\eqref{eq:hqhat} to agree with the plug-in Bayes classifier 
in a small window around the estimated quantile $\widehat{q}$, thus avoiding expensive classification errors when $\widehat{\gamma}(\bx)$ is close to $\widehat{q}$.

\begin{theorem} 
Let $\widehat{h}_q$ be the plug-in classifier defined in \eqref{eq:hqhatmod}. Then for $\widehat{h}_q$  we have that,
$$  \mathcal{E}(\widehat{h}_q) - \mathcal{E}_q  \leq 2\|\widehat\gamma - \gamma \|_1 \leq 2\|\widehat\gamma - \gamma \|_2.$$
\end{theorem}
\begin{proof}

Suppose we consider a point $\bx$, then the excess classification error for the point $\bx$ is given by,
\begin{align*}
 \mathcal{E}^{\bx}(\widehat{h}_q) - \mathcal{E}^{\bx}_q = (1 - 2 \gamma(\bx)) (\widehat{h}_q(\bx) - h_q(\bx)).
\end{align*} 
We will suppose that without loss of generality that the true quantile $q \geq 1/2$ (an identical argument works in case $q < 1/2$). 
Now consider the following cases:
\begin{enumerate}
\item $\gamma(\bx) \geq q$, and $\widehat{\gamma}(\bx) \geq 1/2$
\item $\gamma(\bx) < q$, and $\widehat{\gamma}(\bx) \geq 1/2$
\item $\gamma(\bx) \geq q$, and $\widehat{\gamma}(\bx) < 1/2$
\item $\gamma(\bx) < q$, and $\widehat{\gamma}(\bx) < 1/2.$
\end{enumerate}
In cases (1) and (4), $\widehat{h}_q(\bx) = h_q(\bx)$,
and so the only time we can make an excess error is in cases (2) or (3), and we deal with each of these in turn.

In case (2), we can further consider two cases, (2a) when $\gamma(\bx) \geq 1/2$ and (2b) when $\gamma(\bx) < 1/2$. 
In the first of these cases, the term $(1 - 2 \gamma(\bx)) (\widehat{h}_q(\bx) - h_q(\bx))$ is always $\leq 0$, 
and in case (2b) we have that, 
\begin{align*}
(1 - 2 \gamma(\bx)) (\widehat{h}_q(\bx) - h_q(\bx)) \leq |1 - 2 \gamma(\bx)| \leq 2 |\gamma(\bx) - \widehat{\gamma}(\bx)|,
\end{align*}
since $\gamma(\bx)$ and $\widehat{\gamma}(\bx)$ are on opposite sides of $1/2$. Similarly, 
in case (3), we have that $\gamma(\bx)$ and $\widehat{\gamma}(\bx)$ are once again 
on opposite sides of $1/2$ and the same bound holds. 
Putting these cases together we observe that for any point $\bx$,
\begin{align*}
 \mathcal{E}^{\bx}(\widehat{h}_q) - \mathcal{E}^{\bx}_q \leq 2 |\gamma(\bx) - \widehat{\gamma}(\bx)|,
\end{align*} 
and taking the expectation over the distribution of the point $\bx$ yields the result.
\end{proof}

\section{Variation independence of strength and sharpness} \label{sec:varind}

For $X \sim N(0,1)$ and $\gamma(x)=\Phi(b_0+b_1 x)$ we have
\begin{align*}
\mu &= \int \Phi(b_0+b_1 x) \phi(x) \ dx = \Phi\left( {b_0} \Big/ {\sqrt{1+b_1^2}} \right) \\
q &= \Phi\left\{b_0 + b_1 \Phi^{-1}(1-\mu) \right\} = \Phi\left( b_0 - {b_0 b_1} \Big/ {\sqrt{1+b_1^2}} \right) \\
\psi &= \frac{\E\{\gamma \one(\gamma > q)\}-\mu^2}{\mu(1-\mu)} = \frac{1}{\mu(1-\mu)} \left\{ \int_{\gamma > q} \Phi(b_0+b_1x) \phi(x) \ dx - \mu^2 \right\} \\
&= \frac{1}{\mu(1-\mu)} \left[ \int_{-\Phi^{-1}(\mu)} \Phi\left\{ \left( \sqrt{1+b_1^2} \right) \Phi^{-1}(\mu) + b_1 x \right\} \phi(x) \ dx - \mu^2 \right] .
\end{align*}
The following figure plots $\psi=\psi(\mu,b_1)$ as a function of $b_1$ and $\mu$. This plot allows one to read off what $b_1$ value is needed to ensure given strength and sharpness $(\mu,\psi)$. Then $b_0$ can be obtained using $b_0=\Phi^{-1}(\mu) \sqrt{1+b_1^2}$. 

\begin{figure}[h!]
\spacingset{1}
\centering
\includegraphics[width=.5\textwidth]{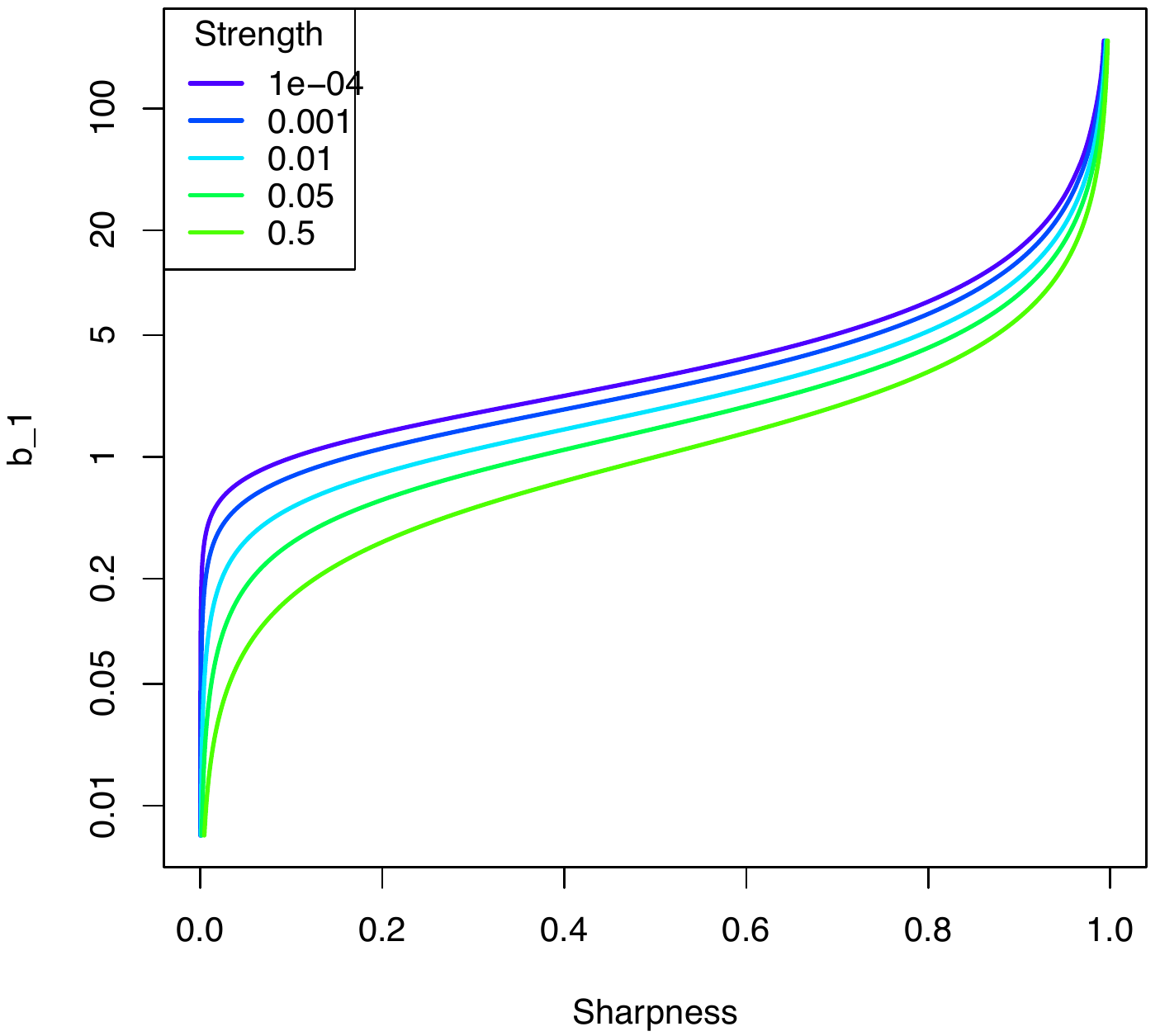}
\caption{Values of $b_1$ in $\gamma(x)=\Phi(b_0+b_1 x)$ needed to ensure given levels of sharpness and strength for $X \sim N(0,1)$.  }
\label{fig:explot}
\end{figure}

\section{Margin condition example} \label{sec:margin}

\begin{figure}[h!]
\centering
\hspace*{-.2in} \includegraphics[width=1.1\textwidth]{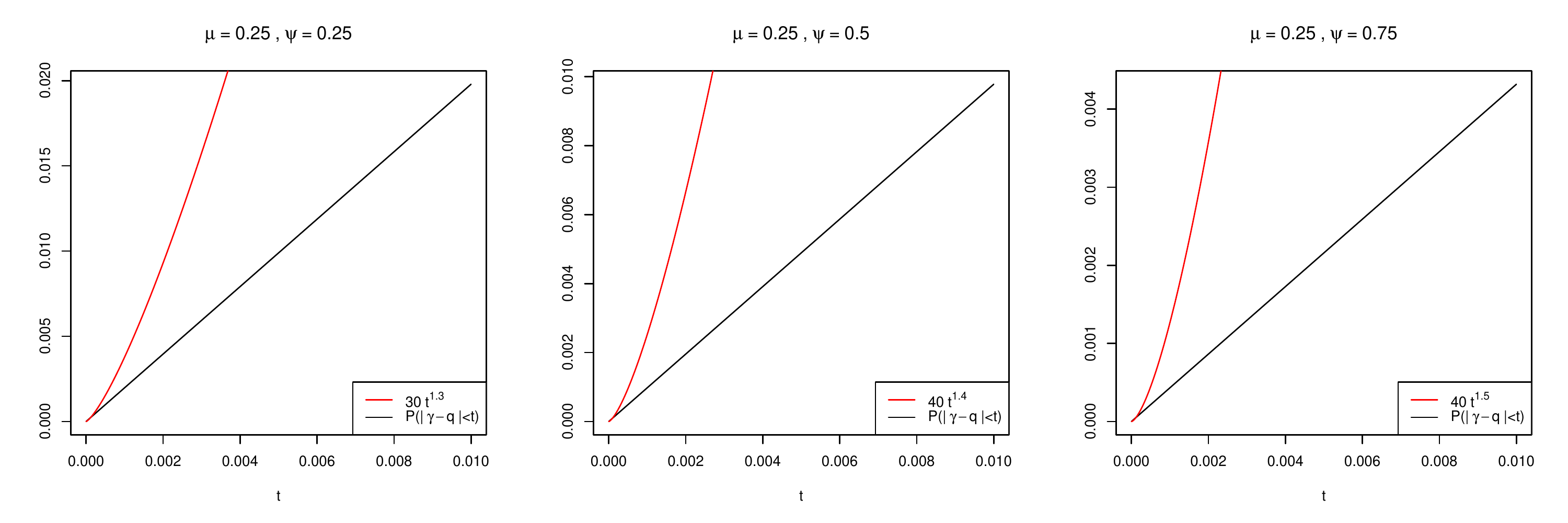}
\caption{$\Pb(| \gamma - q| \leq t)$ for the simulated example from Section \ref{sec:varind}. }
\label{fig:marginplot}
\end{figure}

\noindent In this section, we briefly investigate the margin condition for the example in Appendix~\ref{sec:varind}. Figure \ref{fig:marginplot} shows examples of $(C,\alpha)$ satisfying the margin condition~\eqref{eqn:marginsiva}, in the simulated example where $X \sim N(0,1)$ and $\gamma(x) = \Phi(b_0 + b_1 x)$, for three values of sharpness (the plots are similar when varying strength). A sharp IV with $\psi=0.75$ allows $\alpha = 1.5$. In each case, the margin parameters $(C,\alpha)$ are determined numerically by searching over a small grid of values to determine the best values for which the condition holds.

\section{Equivalence with Youden index} \label{sec:youdenequiv}

Here we show that variance explained $\cov(C,h)/\var(C)$ equals the Youden index. Note
\begin{align*}
\cov(C,h) &= \E(hC) - \mu\E(h) = \E(h \mid C=1) \mu - \mu \{ \E(h \mid C=1) \mu + \E(h \mid C=0) (1-\mu) \} \\
&= \mu (1-\mu) \{ \E(h \mid C=1) - \E(h \mid C=0) \}
\end{align*}
so the equivalence follows by the fact that $\var(C) = \mu(1-\mu)$.

\section{Logit-transformed confidence interval for sharpness}
In this section, we derive the logit-transformed 
confidence intervals for sharpness that we use in our numerical experiments.

\begin{proposition} \label{logit}
Assume the same conditions as in Theorem \ref{psihat}, and let 
$$ \varphi_\psi = \frac{\phi_\mu h_q + q(\phi_\mu - h_q)  - \xi }{(\mu-\mu^2)}  + \frac{ (2\mu\xi - \xi - \mu^2) }{ (\mu-\mu^2)^2 } (\phi_\mu - \mu) $$
denote the efficient influence function for $\psi$. Then
$$ \sqrt{n} \Big\{ \logit(\widehat\psi) - \logit(\psi) \Big\} \indist N\left(0, \var \left\{ \frac{\varphi_\psi(\bO;\boldsymbol\eta)}{\psi(1-\psi)} \right\} \right)  $$
and 
$$ \expit\left[ \logit(\widehat\psi) \pm 1.96 \sqrt{ \widehat\var \left\{ {\varphi_\psi(\bO;\boldsymbol{\widehat\eta})} / {(\widehat\psi-\widehat\psi^2)} \right\} \Big/ {n} }  \right] $$ 
is an asymptotic 95\% confidence interval for sharpness $\psi$ taking values in the unit interval. 
\end{proposition}

\noindent Proposition \ref{logit} follows from the delta method noting that $\partial \logit(\psi)/\partial \psi = 1/(\psi-\psi^2)$, together with the fact that $\expit(\cdot)$ is a monotone transformation.

\section{Simulation code} \label{sec:code}
In this section, we provide all the necessary code to reproduce our simulations in Section~\ref{sec:sims}.

\begin{verbatim}
install.packages("devtools"); library(devtools)
install_github("ehkennedy/npcausal"); library(npcausal)
expit <- function(x){exp(x)/(1+exp(x))}
logit <- function(x){log(x/(1-x))}

# set parameters
set.seed(2000); nsim <- 500; i <- 1
n <- 500; mu <- 0.3; psi <- 0.2; eff <- 0.2
cols <- c("psi","psi.ci1","psi.ci2","ate.lb","ate.ub","ate.ci1","ate.ci2",
          "bhq.lb","bhq.ub","bhq.ci1","bhq.ci2","h0err","hqerr","hserr")
res <- as.data.frame(matrix(nrow=nsim,ncol=length(cols)))
colnames(res) <- cols

# find values that yield set strength/sharpness
bseq <- exp(seq(-2.8,5.5,length.out=10000)); sharpfn <- function(b){ 
  (integrate(function(x){ pnorm(sqrt(1+b^2)*qnorm(mu)+b*x)*dnorm(x) },
  -qnorm(mu),Inf)$value - mu^2) / (mu-mu^2) }; psival <- sapply(bseq,sharpfn)
bval <- bseq[which.min(abs(psi-psival))]; aval <- qnorm(mu)*sqrt(1+bval^2)

for (i in 1:nsim){ print(i); flush.console()
# simulate data
x <- rnorm(n); gamma <- pnorm(aval + bval*x); c <- rbinom(n,1,gamma)
pi <- expit(x); z <- rbinom(n,1,pi); a <- c*z + (1-c)*rbinom(n,1,.5)
y1 <- rbinom(n,1,.5+eff/2); y0 <- rbinom(n,1,.5-eff/2); y <- a*y1 + (1-a)*y0

# estimate effects/strength/sharpness
res1 <- ivlate(y,a,z,cbind(1,x,gamma),nsplits=2, sl.lib=c("SL.glm"))
res2 <- ivbds(y,a,z,cbind(1,x,gamma),nsplits=2, sl.lib=c("SL.glm"))
res[i,1:3] <- res1$res[3,c(2,4,5)]; res[i,4:7] <- res2$res[1,2:5]
res[i,8:11] <- res2$res[2,2:5]
res$h0err[i] <- mean((res2$nuis$gamhat>0.5)!=c)
res$hqerr[i] <- mean(res2$nuis$hq!=c)
res$hserr[i] <- mean((res2$nuis$gamhat>runif(n))!=c)  }
# summarize results
mean(res$h0err); mean(res$hqerr); mean(res$hserr)
mean(res$ate.ub-res$ate.lb); mean(res$bhq.ub-res$bhq.lb)
mean(res$ate.ci1 < eff & eff < res$ate.ci2)
mean(res$bhq.ci1 < eff & eff < res$bhq.ci2)
mean(res$psi-psi); sd(res$psi)
mean(res$psi.ci1 < psi & psi < res$psi.ci2)
\end{verbatim}

\end{document}